\documentclass[11pt,oneside,reqno]{amsart}
\usepackage[margin=1in]{geometry}
\usepackage{amsfonts,amsthm}
\usepackage{amssymb}
\usepackage[dvips]{graphics}
\usepackage{epsfig}
\usepackage{euscript}
\usepackage{color}
\usepackage{bbm}

\usepackage{times}
\usepackage{enumerate}
\usepackage{bbold}
\usepackage{hyperref}

\definecolor{purple}{rgb}{.5,0,1}
\definecolor{orange}{rgb}{1,.5,0}
\definecolor{pink}{rgb}{1,0,.5}

\numberwithin{equation}{section}
\newtheorem{theorem}{Theorem}[section]
\newtheorem{proposition}[theorem]{Proposition}
\newtheorem{lemma}[theorem]{Lemma}
\newtheorem{corollary}[theorem]{Corollary}

\DeclareMathOperator{\supp}{supp}
\DeclareMathOperator{\tr}{tr}

\DeclareMathOperator{\dist}{dist}

\DeclareMathOperator{\Rea}{Re}
\DeclareMathOperator{\Ima}{Im}

\newcommand\R{\mathbb R}
\newcommand\N{\mathbb N}
\newcommand\C{\mathbb C}
\newcommand\Z{\mathbb Z}

\newcommand\cF{\mathcal{F}}

\newcommand\e{\mathrm{e}}

\renewcommand\P{\mathbb P}
\newcommand\E{\mathbb E}
\newcommand\cE{\mathcal{E}}

\newcommand\cN{\mathcal{N}}

\newcommand\cV{\mathcal{V}}

\newcommand\cX{\mathcal{X}}

\newcommand\cS{\mathcal{S}}
\newcommand\cQ{\mathcal{Q}}
\newcommand\cR{\mathcal{R}}

\newcommand\tW{\tilde{W}}

\newcommand\eps{\varepsilon}

\renewcommand{\d}{\mathrm{d}}

\newcommand{\pr}{\prime}

\newcommand\what{\widehat}

\newcommand\beq{\begin{equation}}
\newcommand\eeq{\end{equation}}

\newcommand{\abs}[1]{\left\lvert #1 \right\rvert}
\newcommand{\norm}[1]{\left\lVert #1 \right\rVert}
\newcommand{\scal}[1]{\left\langle #1 \right\rangle}
\newcommand{\set}[1]{\left\{ #1 \right\}}
\newcommand{\pa}[1]{\left( #1 \right)}

\newcommand{\fl}[1]{\left\lfloor #1 \right\rfloor}

\newcommand\La{\Lambda}

\newcommand{\eq}[1]{\eqref{#1}}

\newcommand{\up}[1]{^{(#1)}}

\newcommand{\qtx}[1]{\quad\text{#1}\quad}
\newcommand{\mqtx}[1]{\; \ \text{#1}\; \  }
\newcommand{\sqtx}[1]{\;\text{#1}\;}

\begin{document}

\title[Localization in the XXZ Chain]{Many-body localization  in the droplet spectrum\\ of the random XXZ quantum spin chain}

\author[[A. Elgart]{Alexander Elgart}
\address[A. Elgart]{Department of Mathematics\\
Virginia Tech\\
Blacksburg, VA, 24061, USA}
\email{aelgart@vt.edu}
\author[A. Klein]{Abel Klein}
\thanks{A.\ K.\ was supported in part by the NSF under grant DMS-1001509}
\address[A. Klein]{University of California, Irvine\\ Deparment of Mathematics\\ Irvine, CA, 92697-3875, USA}
\email{aklein@uci.edu}
\author[G. Stolz]{G\"unter Stolz}
\address[G. Stolz]{Department of Mathematics\\
University of Alabama at Birmingham\\
Birmingham, AL 35294-3875, USA}
\email{stolz@uab.edu}

\date{\today}

\begin{abstract}
We study many-body localization properties of the disordered XXZ spin chain in the Ising phase. Disorder is introduced via a random magnetic field in the $z$-direction. We prove a strong form of dynamical exponential clustering for eigenstates  in the droplet spectrum: For any pair of local observables separated by a distance $\ell$, the sum of the associated correlators over these states decays exponentially  in $\ell$, in expectation. This exponential clustering persists under the time evolution in the  droplet spectrum.  Our result applies to the large disorder regime as well as to the strong Ising phase at fixed disorder, with bounds independent of the support of the observables.
\end{abstract}

\maketitle

\setcounter{tocdepth}{1}
\tableofcontents


\section{Introduction}

\subsection{Many-Body Localization} 

Understanding the structure and complexity of the eigenstates of   ``typical"
local Hamiltonians is one of the central problems in Condensed Matter Physics 
and Quantum Complexity Theory. The concept of disorder induced localization 
has been introduced in the seminal work  of Anderson \cite{And}, who 
suggested a mechanism responsible for the absence of diffusion of waves in 
disordered media. This mechanism is  well understood  by now in the single
particle case, both physically and mathematically.  
  In random Schr\"odinger operators, localization  manifests itself  as pure point spectrum, with the corresponding eigenvectors  exhibiting exponentially fast spatial 
decay, for 
almost all configurations of the environment.

 It turns out that many manifestations of single-particle Anderson localization 
remain valid  if one consider a fixed number of interacting 
particles, e.g., \cite{CS,AW,KN}. The methodology used in these works is unfortunately inadequate to study the thermodynamic limit of an electron gas in a random 
environment, i.e., an infinite volume limit in which the number of electrons 
grows proportionally to the volume. In the landmark paper \cite{Baskoetal} it was suggested that some hallmarks of localization indeed survive the passage to a true many-body system. This has sparked extensive efforts in the physics community to understand this phenomenon, known as many-body localization (MBL), see,  e.g., \cite{Bardarsonetal,BauerNayak,BH,OganesyanHuse,Abaninetal,VoskAltman,Znidaricetal}. This compilation, far from being comprehensive,  lists only a few salient works closely related to the current paper. Definitions and heuristic arguments that lay a general foundation for an analytical approach to MBL were introduced in \cite{HastingsMBL},  based on earlier work on gapped systems  (e.g., \cite{HastingsAL}), by identifying their possible counterparts in mobility gapped systems.

Investigating the combined effects of disorder and interactions on large 
quantum systems is difficult, because the many-body quantum states  are 
extremely complicated objects. Even an approximate description of such a  state, 
in general, requires the introduction of an exponential number of parameters 
(in terms of the system's size). As a consequence, the basic questions related to the behavior of  disordered many-body quantum systems are still subject of debate in the physics community. However, a clearer phenomenological picture has by now been drawn in the so called {\it fully many-body localized} regime, e.g., \cite{Chandranetal,Huse,Imbrie,AbaninFMBL}. 

Mathematical proofs of accepted MBL characteristics such as zero-velocity Lieb-Robinson bounds, rapid decay of correlations, as well as area laws for the bipartite entanglement of eigenstates, have generally been restricted to quasi free systems where MBL properties can be reduced to Anderson localization of an effective one-particle Hamiltonian. Examples are the XY spin chain in random transversal field (see \cite{ARNSS} for a recent review), the disordered Tonks-Girardeau gas \cite{SeiringerWarzel}, and systems of quantum harmonic oscillators \cite{NSS1,NSS2}. The mean field theory in the Hartree-Fock approximation for random systems was studied in \cite{Ducatez}.

Very few rigorous results exist for models where this type of reduction is not possible, and their scope is rather limited.  One can  mention the exponential clustering property for the ground state of the Andr\'e-Aubry quasi-periodic model  \cite{Mas1,Mastropietro}.

A model which has received considerable attention in the physics literature is the XXZ chain in  a random field. In particular, numerical evidence suggests that this model exhibits a many-body localization-delocalization transition in the weak  disorder regime, e.g.\ \cite{Agarwaletal, Luitzetal, PalHuse}. What makes this model accessible to rigorous analysis is that it is particle number preserving in the  case of a field in $z$-direction.  This allows  the reduction to an infinite system of discrete $N$-body Schr\"odinger operators on the fermionic subspaces of $\Z^N$ \cite{NachtergaeleStarr, NaSpSt, FS}. In particular, for the Ising phase of the XXZ chain, in the absence of an external field, the low energy states above the ground state are characterized by a {\it droplet regime}. In this regime, as can be explained by an attractive particle interaction in the associated Schr\"odinger operators, spins form a droplet, i.e., a single cluster of down spins (in the normalization chosen below) in a sea of up spins.

In this paper we prove that the XXZ chain in a random field  exhibits one of the expected hallmarks of localization, namely a strong form of the exponential clustering property, in the droplet spectrum. In concrete terms, we show  exponential clustering of the averaged correlations of local observables in all eigenstates  with energies in the droplet spectrum. We actually establish  a stronger \emph{dynamical} exponential clustering property:   this exponential clustering is preserved under the time evolution in the droplet spectrum.
While our analysis relies on establishment of the localization properties of the associated Schr\"odinger operators, the main result requires a set of new ideas on how to translate these properties to the global theory for the spin chain.

Our proof works in the Ising phase for two distinct regimes, the case of random field with large disorder  (described by a parameter $\lambda >0$),  as well as  in the strong Ising phase (described by a parameter $\Delta>1$).  In fact, our proof works in the parameter region described by 
 \beq\label{eq:paramreg}
\pa{\Delta,\lambda}\in   (1,\infty)\times (0,\infty)\qtx{satisfying} \lambda\sqrt{\Delta -1}\min \set{1, (\Delta -1)}\ge K,\eeq  
where the finite constant $ K>0  $ is system size independent.
The underlying physical mechanism  is that the droplets formed by the spins can be considered as single quasi particles, which become localized in the presence of disorder. Crucial for the rigorous analysis is that this can be quantified with bounds independent of the size of the droplets.

In a sequel to this  work \cite{EKS}, we investigate further dynamical manifestations of localization in the droplet spectrum of the random XXZ chain, including non-spreading of information, zero-velocity Lieb-Robinson bounds, and general dynamical clustering.

 After circulating an early draft of this work we learned about work in progress   concerning the $N$-body Schr\"odinger operators associated with the random XXZ model,  now available in \cite{BeaudWarzel}. 

\vspace{.3cm} 

\noindent {\bf Acknowledgements:} We are indebted to Bruno Nachtergaele for suggesting this  investigation of the random XXZ chain and for many useful discussions.

\subsection{The Main Result} \label{sec:mainresult}

The infinite  XXZ chain in  a random field is given by the (formal) Hamiltonian 
\beq \label{infXXZ}
H=H_\omega=H_0+\lambda B_\omega,\quad H_0=\sum_{i\in\Z}h_{i,i+1},\quad B_\omega=\sum_{i\in\Z} \omega_i \mathcal{N}_i,
\eeq
acting on quantum spin configurations on the one-dimensional lattice $\Z$. The local next-neighbor Hamiltonian $h_{i,i+1}$ is given by
\beq 
h_{i,i+1}=\tfrac{1}{4}\pa{I-\sigma_i^z\sigma_{i+1}^z}-\tfrac{1}{4\Delta}\pa{\sigma_i^x\sigma_{i+1}^x+\sigma_i^y\sigma_{i+1}^y},
\eeq
where $\sigma^{x,y,z}$ are the standard Pauli matrices and $\Delta$ is a positive parameter. Also, 
\beq \label{eq:number}
\mathcal{N}_i = \tfrac{1}{2} (1-\sigma_i^z) = \begin{pmatrix} 0 & 0 \\ 0 & 1 \end{pmatrix}_i
\eeq 
is the projection onto the down-spin state (or local number operator) at site $i$. The positive parameter $\lambda$ describes the strength of the disordered longitudinal magnetic field $B_\omega$. We choose the random parameters $\omega = (\omega_i)_{i\in\Z}$ as  independent identically distributed random variables with distribution $\mu$, assuming that supp$\,\mu=[0,\omega_{max}]$ for some finite $\omega_{max}$ and that $\mu$ is absolutely continuous with bounded density $\rho$. Thus the random field $B_{\omega}$ is non-negative. We have normalized both $H_0$ and $B$ so that the ground state energy of $H$ is $E_0=0$, independent of the random parameters $\omega$, with ground state given by the all-spins up configuration (or {\it vacuum vector}).

As special cases one gets the Heisenberg chain for $\Delta=1$ and the Ising chain in the limit $\Delta\to\infty$. Here we  consider the Ising phase of the XXZ chain, i.e.,  we assume $\Delta>1$.  The factor $\frac1 4$ in our normalization is reminiscent of the fact that many physics papers use the spin matrices $S^{x,y,z} = \frac{1}{2}\sigma^{x,y,z}$ in the definition of the Hamiltonian. 

We will generally work with restrictions of $H$ to finite intervals $[-L,L]$, $L\in \N$, i.e., 
\beq \label{finiteXXZ}
H^{(L)} = \sum_{i=-L}^{L-1} h_{i,i+1} + \lambda\sum_{i=-L}^L \omega_i \mathcal{N}_i + \beta (\mathcal{N}_{-L} + \mathcal{N}_L),
\eeq
which is self-adjoint on $\otimes_{-L}^L \C^2$. In the {\it droplet boundary term} $\beta (\mathcal{N}_{-L} + \mathcal{N}_L)$ (compare \cite{NachtergaeleStarr}) we choose $\beta \ge \frac{1}{2}(1-\frac{1}{\Delta})$. 

The choice of the boundary condition in (\ref{finiteXXZ}) is due to our methods of proof, as it gives us a convenient positivity property. While it is likely that the case $\beta< \frac{1}{2}(1-\frac{1}{\Delta})$ can be covered by similar methods, we heavily use particle number conservation of the model (see Section~\ref{sec:fixedN}), leaving more general boundary terms out of our reach at the moment.

As we will see from the more rigorous definition of the infinite volume Hamiltonian $H$ as a random operator in Section~\ref{sec:fixedN}, its spectrum is given with probability one by $\sigma(H) =\{0\} \cup [1-\frac{1}{\Delta},\infty)$, see also \cite{FS}. Our main goal here is to show that $H$ is many-body localized in the {\it droplet spectrum} $I_1 =[1-\frac{1}{\Delta}, 2(1-\frac{1}{\Delta}))$ if $\Delta$ is sufficiently large or $\lambda$ is sufficiently large. This will be expressed in terms of exponential clustering of the eigenstates of the finite chain $H^{(L)}$ for energies in $I_1$, uniform in $L$.

 We will see that for each $L$  the spectrum of $H^{(L)}$ is almost surely simple, so that its normalized eigenvectors can be labeled as $\psi_E$, $E\in \sigma(H^{(L)})$. 
Given a finite subset $J \subset [-L,L]\cap \Z $, a local observable $X$  with support $J$ is an operator on $\otimes_{j\in J} \C^2$, considered as an operator on $\otimes_{-L}^L \C^2$ by acting as the identity on spins not in $J$; we write $J= \supp X$.
A useful concept that quantifies  how close a many body state $\psi$ is to a product state is its correlator:
\beq\label{eq:correlful}
R_{X,Y}(\psi):=\abs{\scal{\psi,XY\psi}-\scal{\psi,X\psi}\scal{\psi,Y\psi}},
\eeq
where  $X$ and $Y$ are local observables.  The rapid decay of correlators for  gapped ground states was first discovered by Hastings \cite{HastingsAL}, and led  to  the introduction of  definitions and heuristic arguments in  \cite{HastingsMBL} that lay a general foundation for an analytical approach to MBL.

Localization of the finite chain $H^{(L)}$ in an interval  $I$ actually leads to the rapid decay of
 the sum of the correlators of {\it all} eigenstates  with eigenvalues in $I$, \ 
 $\sum_{E\in \sigma(H^{(L)}) \cap I}   {R_{X,Y}}(\psi_E) $.  To study this and more general correlators, note that
the correlator of a simple  eigenvector $\psi_E$ can   be rewritten as
\beq
R_{X,Y}(\psi_E)=\abs{\tr \pa{P\up{L}_E X\bar P\up{L}_E Y P\up{L}_E}},\eeq
where  $P\up{L}_E=  P\up{L}_{\set{E}}$, with   $P\up{L}_F= \chi_F(H\up{L})$ and $\bar P\up{L}_F =1 -  P\up{L}_F$ for $F\subset \R$.
This suggests  defining the correlator of an energy  set $F\subset \R$ by
\beq\label{eq:correlfulK}
R_{X,Y}(F)= \abs{\tr \pa{P\up{L}_F X\bar P\up{L}_FY P\up{L}_F}}=\abs{\tr \pa{P\up{L}_F X Y P\up{L}_F}- \tr \pa{P\up{L}_F XP\up{L}_F Y P\up{L}_F}}.
\eeq
 In addition, given    a finite energy  interval $I$, we let $\cF_I$ denote the collection of all partitions $\cF=\set{F_q}_{q\in \cQ}$ of $I$ into finitely many disjoint intervals (of any kind),
and define
\beq\label{defcRXY}
\cR_{X,Y}(I)= \sup_{\cF \in \cF_I}R_{X,Y}(\cF),\sqtx{where}R_{X,Y}(\cF)=\sum_{q\in \cQ} R_{X,Y}(F_q)\sqtx{for} \cF=\set{F_q}_{q\in \cQ}\in\cF_I.
\eeq
Note that, since $\sigma(H^{(L)}) \cap I$ is a finite set, $\cR_{X,Y}(I)$ is a measurable function.

In
one particle localization,  the localization properties of the eigenfunctions in an energy interval $I$ (closely related to the concept of exponential clustering in many body systems) lead to  dynamical localization (non spreading of the wave packets under the time evolution) in the same energy window. As we shall see, this phenomenon persists in the many body setting as well, in the form of  exponential clustering in an energy interval under the Heisenberg dynamics restricted to the same interval. Specifically, the time evolution of a local  observable $X$ in the energy window I for a finite chain $H^{(L)}$ in an interval $I$ is given by
\beq
\tau_t^{I}(X)=\pa{\tau_t^{I}}^{(L)}\pa{X}=\e^{itH^{(L)}_I}X\e^{-itH^{(L)}_I}, \qtx{with} H^{(L)}_I=P^{(L)}_I H^{(L)}.
\eeq

Our main result is given in the following theorem, where $\E$  denotes the expectation with respect to the random variables $\omega$. We  fix  $0<\delta< 1$, and  let
\beq \label{dropspec}
I_{1,\delta} = \left[ 1- \tfrac{1}{\Delta}, (2-\delta)\big(1-\tfrac{1}{\Delta}\big) \right],
\eeq
meaning that we  set a fixed, but arbitrarily small, distance from the upper end of the droplet spectrum $I_1$.

\begin{theorem}[Dynamical exponential clustering in the droplet spectrum]\label{thm:expclustering} 
There exists a constant $K>0$ with the following property: If  the parameters $\pa{\Delta,\lambda}$ are in the region described by  \eqref{eq:paramreg}, there exist constants 
$C<\infty$ and $m>0$ such that, letting $I=  I_{1,\delta}$, for all local observables $X$ and $Y$ with  $\max \supp X < \min \supp Y$ (or vice versa) we have, uniformly in $L$,  
\beq \label{eq:expclustering0}
\E \pa{ \sup_{t\in \R} \cR_{\tau^{I}_t\pa{X},Y}(I)}\le C \|X\| \|Y\| \e^{-m \dist\pa{\supp X,\,  \supp Y}}.
\eeq
In particular, we have the special cases
\beq \label{eq:expclustering}
\E \pa{ \sup_{t\in \R}  \sum_{E\in \sigma(H^{(L)}) \cap I}   {R_{\tau^{I}_t\pa{X},Y}}(\psi_E) }\le C \|X\| \|Y\| \e^{-m \dist\pa{\supp X,\,  \supp Y}}
\eeq
 and 
\beq \label{eq:expclustering9998}
\E \pa{ \sup_{t\in \R}    {R_{\tau^{I}_t\pa{X},Y}}(I) }\le C \|X\| \|Y\| \e^{-m \dist\pa{\supp X,\,  \supp Y}}.
\eeq 
\end{theorem}

The left hand side  in \eq{eq:expclustering}  is the average of the supremum in time of the sum of the correlators of {\it all} eigenstates  with eigenvalues in the droplet spectrum $I_{1,\delta}$. A  special case of Theorem~\ref{thm:expclustering}  is exponential clustering of {\it stationary} correlations (corresponding to the choice  $t=0$ in its statement). The constants $K,C,m$ depend on  the  distribution $\mu$ and on  the parameter $\delta$.
Note that  the bound (\ref{eq:expclustering0}) depends only on the distance between the supports of $X$ and $Y$, without pre-factors depending on the sizes of these supports-- a particularly strong feature of our model and the methods used to prove the result. 
Theorem 1.1 extends exponential clustering bounds previously only known for stationary ground state correlations of gapped systems (e.g., \cite{NachtergaeleSims}), and also covers dynamical correlations at arbitrary time.

In view of a result by Brandao and Horodecki  \cite{BH}, this behavior strongly suggests that droplet states satisfy an area law for the bipartite entanglement as well. However, an additional analysis is required to properly account for the use of disorder averages in our context.

\section{Strategy}

An important special case of the local observables $X$ and $Y$ are the local number operators $\mathcal{N}_i$ in \eqref{eq:number}.  Their correlators in  a normalized  eigenvector $\psi_E$ of  $H^{(L)}$ satisfy
 \begin{equation} \label{Ntomany} 
R_{\mathcal{N}_i,\mathcal{N}_j}(\psi_E) = |\langle \psi_E, \mathcal{N}_i (I-|\psi_E \rangle \langle \psi_E|) \mathcal{N}_j \psi_E \rangle | \le \norm{\mathcal{N}_i\psi_E}\norm{\mathcal{N}_j\psi_E} .
\end{equation}
The  next theorem establishes Theorem~\ref{thm:expclustering} for the special case of $X= \mathcal{N}_i$, $Y=\mathcal{N}_j$, and  $t=0$.

Given an interval $I$, we set  
$
G_I= \set{g:\R \to\C\mqtx{Borel measurable,} \abs{g}\le  \chi_I}.
$
We also let  $\norm{ \ }_1$ denote  the trace norm.

\begin{theorem}[Exponential localization in the droplet spectrum] \label{thm:efcor} 
There exists a constant $K>0$ with the following property: If  the parameters $\pa{\Delta,\lambda}$ are in the region described by  \eqref{eq:paramreg}, there exist constants 
$C<\infty$ and $m>0$ such that
\beq \label{eq:efcor5}
 \E\pa{ \sum_{E\in \sigma(H^{(L)}) \cap I_{1,\delta}} \norm{\mathcal{N}_i\psi_E}\norm{\mathcal{N}_j\psi_E}} \le C e^{-m|i-j|} \mqtx{for all} -L \le i, j \le L,
\eeq
  and
\beq \label{eq:efcor59}
 \E\pa{\sup_{g \in G_{ I_{1,\delta}}} \norm{\mathcal{N}_i g(H^{(L)}) \mathcal{N}_j}_1} \le C e^{-m|i-j|} \mqtx{for all}  i, j \in[-L, L],
\eeq
uniformly in $L$.
\end{theorem}

 Note that  \eq{eq:efcor59} is an immediate consequence of \eq{eq:efcor5}, because, 
since the spectrum of $H^{(L)}$ is almost surely simple,  we have 
$
\norm{\mathcal{N}_i P_E^{(L)} \mathcal{N}_j}_1= \norm{\mathcal{N}_i\psi_E}\norm{\mathcal{N}_j\psi_E}$.

Theorem~\ref{thm:efcor} will be pivotal in the proof of   Theorem~\ref{thm:expclustering} in its general form (see Section~\ref{sec:expclustering} for details), and  can be viewed as our main technical result. 

\subsection{Results for Fixed Particle Number} \label{sec:fixedN}

The crucial property of the XXZ chain to be exploited  for the proof of Theorem~\ref{thm:efcor}  is particle number conservation.  Let
\beq 
\mathcal{N}^{(L)} = \sum_{i=-L}^L \mathcal{N}_i
\eeq
be the total (down) spin number operator. Then $[H^{(L)}, \mathcal{N}^{(L)}]=0$, and  $\mathcal{N}^{(L)}$ has eigenvalues $\lambda_N=N$, $N=0,1,\ldots,2L+1$, with eigenspaces  spanned by all the spin basis states with $N$ down spins, called the $N$-particle sector (or $N$-magnon sector, as one may equivalently use the eigenspaces of the total magnetization operator $S^z = \sum_i \sigma_i^z$).

The restriction of the XXZ chain to the $N$-particle sector is unitarily equivalent to an $N$-body discrete Schr\"odinger operator restricted to the fermionic subspace. Due to working with the one-dimensional XXZ chain, this operator can also be directly expressed as a Schr\"odinger-type operator over the (induced) subgraph of ordered lattice points $\mathcal X_N=\{x=(x_1,\ldots,x_N)\in \Z^N:\ x_1<x_2<\ldots<x_N\}$ of $\Z^N$, and, in the finite volume case, over $\mathcal{X}_N^{(L)} = \{x\in \Z^N:-L \le x_1 < \ldots < x_N \le L\}$.

For the infinite volume case this is carried out in detail in \cite{FS}. The argument there can be adjusted to the finite volume case \eqref{finiteXXZ}, including the boundary condition, which we need here. As \cite{FS} uses adjacency operators while we will consider graph Laplacians, some care is needed in this derivation due to the fact that the graphs $\mathcal{X}_N^{(L)}$ have non-constant nearest neighbor degree.  As a result, 
the restriction of $H^{(L)}$ to the $N$-particle sector is unitarily equivalent to the self-adjoint operator on $\ell^2(\mathcal{X}_N^{(L)})$ given by
\beq\label{eq:H_N}
H_N^{(L)}=-\tfrac{1}{2\Delta}\mathcal{L}_N^{(L)}+\pa{1-\tfrac{1}{\Delta}}\tilde W+\lambda V_\omega + \left(\beta-\tfrac{1}{2}(1-\tfrac{1}{\Delta})\right) \chi^{(L)}.
\eeq
Here  $\mathcal{L}_N^{(L)}$ is the graph Laplacian on $\mathcal X_N^{(L)}$,
\beq\label{eq:lapl} 
\pa{\mathcal{L}_N^{(L)}\psi}(x)=\sum_{\substack{y\in \mathcal X_N^{(L)}\\ |x-y|=1}} (\psi(y)-\psi(x)),
\eeq
with $\abs{x-y}:=\sum_{i=1}^N\abs{x_i-y_i}$ denoting graph distance,
\beq
 \tilde W(x) = N-\#\pa{j:\ x_{j+1}=x_j+1} = 1+\#\pa{j:\ x_{j+1} \neq x_j + 1}, 
 \eeq
and the $N$-body random potential $V_\omega$ is given by
\beq
\pa{V_\omega\psi}(x)=\pa{\sum_{j=1}^N\omega_{x_j}}\psi(x).
\eeq  
Finally, in the last term of (\ref{eq:H_N}) we have $\chi^{(L)}=\chi_{-L}+ \chi_{L}$, where $\chi_{-L}$ and $\chi_{L}$ denote the indicator functions of the left and right  boundaries
\beq 
\{(x_1,\ldots,x_N)\in \mathcal{X}_{N}^{(L)}: x_1=-L\} \qtx{and}
 \{(x_1,\ldots,x_N)\in \mathcal{X}_{N}^{(L)}: x_N=L\}
\eeq
of $\mathcal{X}_N^{(L)}$ within $\mathcal{X}_N$. The exact value  $\beta -\frac 1 2(1-\frac 1\Delta)$ of the pre-factor in (\ref{eq:H_N}) is due to the fact that part of the boundary term $\beta (\mathcal{N}_{-L} + \mathcal{N}_L)$ in (\ref{finiteXXZ}) is absorbed into the restricted graph Laplacian $-\tfrac{1}{2\Delta}\mathcal{L}_N^{(L)}$. 
This explains our assumption $\beta \ge \frac 1 2(1-\frac 1\Delta)$, as this assures that the last term in (\ref{eq:H_N}) is non-negative.  

 These considerations  yield the unitary equivalence
\beq \label{eq:magdec}
H^{(L)} \cong \bigoplus_{N=0}^{2L+1} H_N^{(L)}, 
\eeq 
where one identifies the standard basis vectors  $\phi_x := \delta_x\in \ell^2(\mathcal{X}_N^{(L)})$, $x\in \mathcal{X}_N^{(L)}$, of  up-down spin configurations over $[-L,L]$, with down-spins in the positions $x_1<\ldots<x_N$ and up-spins elsewhere. Here $H_0^{(L)}$ denotes the zero operator on a one-dimensional Hilbert space, representing the all up-spins ground state of $H^{(L)}$.

The identity (\ref{eq:magdec}) also provides a convenient way to rigorously define the infinite XXZ Hamiltonian as the direct sum $H = \oplus_{N=0}^{\infty} H_N$, where 
\beq\label{HNinfty}
H_N=-\tfrac{1}{2\Delta}\mathcal{L}_N+\pa{1-\tfrac{1}{\Delta}}\tilde W+\lambda V_\omega,
\eeq
as an operator on $\ell^2(\mathcal{X}_N)$. Boundedness of the random variables $\omega_i$ assures that each $H_N$ is bounded and self-adjoint. Their norms grow linearly in $N$, so that $H$ becomes an unbounded self-adjoint operator on the direct sum of these Hilbert spaces. On infinite spin configurations with finitely many down-spins, which form an operator core of the direct sum, $H$ acts formally  by the expression in \eq{infXXZ}.

We will be interested in the droplet regime of the XXZ chain. To describe this, first set $V_{\omega}=0$ and consider the infinite volume unperturbed operators 
\beq
H_{N,0}=-\tfrac{1}{2\Delta}\mathcal{L}_N+\pa{1-\tfrac{1}{\Delta}}\tilde W
\eeq 
in $\ell^2(\mathcal{X}_N)$. These operators are purely absolutely continuous due to their invariance under the translations $T_N(x_1,\ldots,x_N) = (x_1+1,\ldots, x_N+1)$ on $\mathcal{X}_N$, interpreted as a shift of the center of mass (see \cite{FS} for a review of this and the other properties of $H_{N,0}$ discussed in the following). $H_{N,0}$ can be explicitly diagonalized via the Bethe ansatz for arbitrary $N$. In particular, in the Ising phase $\Delta>1$ the term $(1-\frac 1\Delta) \tilde{W}$ represents an attractive next-neighbor interaction (each pair of particles occupying neighboring sites lowers the energy by $(1-\frac1\Delta)$). This leads to a ``droplet band'' $\delta_N$ at the bottom of the spectrum in each $N$-particle sector. The corresponding generalized eigenfunctions are concentrated at the one-dimensional ``edge'' 
\beq \label{edge}
\mathcal{X}_{N,1} = \{x=(x_1, x_1+1, \ldots, x_1+N-1): x_1 \in \Z\}
\eeq 
of $\mathcal{X}_N$, along which they are quasi-periodic Bloch waves, and decay exponentially into all $N-1$  bulk directions of $\mathcal{X}_N$ (representing the separation distance of particle clusters). Thus droplet states are $N$-particle states with all particles packed into $N$ neighboring sites, up to exponentially small tails. In the spin chain this corresponds to states which are exponentially close to a single droplet of $N$ neighboring down-spins in a sea of up-spins, compare \cite{NachtergaeleStarr}.

The droplet bands are explicitly given by
\beq
\delta_N = \left[ \tanh(\rho) \cdot \frac{\cosh(N\rho)-1}{\sinh(N\rho)}, \tanh(\rho) \cdot \frac{\cosh(N\rho)+1}{\sinh(N\rho)} \right],
\eeq
where $\cosh(\rho) = \Delta$. They form a decreasing sequence of intervals, the first few given by 
\beq
\delta_1 = [1-\tfrac1\Delta, 1+\tfrac1\Delta], \quad \delta_2 = [1-\tfrac1{\Delta^2}, 1], \quad  \delta_3 = [1- \tfrac1{2\Delta^2-\Delta}, 1- \tfrac1{2\Delta^2+\Delta}],
\eeq
 which for $N\to\infty$ contract monotonically into the single point $\delta_{\infty} = \set{\sqrt{1-\frac1{\Delta^2}}}$. For $\Delta>3$ these bands are strictly separated from the rest of the spectrum of $H$, which consists of the ground state energy $0$ and an infinite number of additional bands of bulk spectrum, corresponding to non-trivial scattering channels of the $N$-body operator $H_{N,0}$, all contained in $[2(1-\frac 1\Delta), \infty)$.

Adding the positive random field $V_{\omega}$ will enlarge the spectral bands upwards. Indeed, the $H_N$ are ergodic with respect to the shifts $T_N$ and have almost sure spectrum $\Sigma_N = \sigma(H_{N,0}) +[0,N\omega_{max}]$. Thus, as $N$ is arbitrarily large, in the almost sure spectrum $\Sigma = \{0\} \, \cup\,\bigcup_{N=1}^{\infty} \Sigma_N = \{0\} \cup [1-\frac{1}{\Delta}, \infty)$ of $H$ all spectral gaps (with the exception of the ground state gap) will be closed. Still, as the random potential $V_{\omega}$ is non-negative, one expects that all generalized eigenfunctions to energies in the {\it droplet spectrum} $[1-\frac 1\Delta, 2(1-\frac 1\Delta))$ of $H$ will remain localized along the edge $\mathcal{X}_{N,1}$ of the graph. Rigorously establishing this, with bounds uniform in the particle number $N$, will provide us with the main technical ingredient  for the proof of Theorem~\ref{thm:efcor}. 

Technically, we will accomplish the latter by proving localization of a suitable form of many-body eigenfunction correlators for the finite volume operators $H_N^{(L)}$ in the droplet spectrum.   
These are defined as follows:
Given a finite interval $I\subset \mathbb{R}$ and a pair of indices $i,j\in\Z$, the corresponding eigenfunction correlator  is given by 
\beq \label{eq:Q_I}
 Q_N^{(L)}(i,j;I) = \sum_{E\in\sigma \pa{H^{\pa{L}}_N} \cap I}\norm{Q_i\,P_{E}\,Q_j}_1, 
  \eeq
where $\|\cdot\|_1$ denotes the trace class norm, $P_{E}$ denote the spectral projection of $H^{(L)}_N$ onto $E$, and $Q_{i} = Q_{i}^{(N,L)}$ is the indicator function of 
\beq \label{eq:setSx}
S_{\{i\}}  =S_{\{i\}}^{(N,L)}  := \{x\in \mathcal{X}_N^{(L)}: x_j = i \mbox{ for some }j\in\{1,\ldots,N\}\},
\eeq
i.e., the set of  all those lattice sites at which the random potential depends on the random variable $\omega_i$.

As we already mentioned in Section~\ref{sec:mainresult}, the spectrum of  $H^{(L)}$ is almost surely simple. It allows us  to label normalized eigenvectors as $\psi_E$, $E\in \sigma(H^{(L)})$.  Since each eigenvector lies in a fixed $N$-particle sector, we have $\norm{Q_i\,P_{E}\,Q_j}_1= \norm{Q_i\psi_E}\norm{Q_j\psi_E}$ almost surely. We remark that $Q_i=Q_i^{(N,L)}$ is the restriction of the local number operator $\mathcal{N}_i$  to the $N$-particle sector, as can be seen by the action of $\mathcal{N}_i$ on the product basis vectors 
\beq \label{eq:prodbasis} 
e_{\alpha} = e_{\alpha_{-L}} \otimes \cdots \otimes e_{\alpha_L}, \quad \alpha \in \{0,1\}^{\set{-L,-L+1,\ldots,L}},
\eeq
where $e_0= (1,0)^t$ and $e_1=(0,1)^t$.

It follows that, almost surely,
\beq\label{sumNsumQ}
\sum_{N=1}^{\infty}  Q_N^{(L)}(i,j;I) = 
\sum_{E\in \sigma(H^{(L)}) \cap I} \norm{\mathcal{N}_i\psi_E}\norm{\mathcal{N}_j\psi_E}.
\eeq
Thus  the estimate \eq{eq:efcor5} can be reformulated as
\beq \label{eq:efcor}
\sum_{N=1}^{\infty} \E(Q_N^{(L)}(i,j;I_{1,\delta}))\le C e^{-m|i-j|} \mqtx{for all} -L \le i, j \le L,
\eeq
a more convenient form of the bound that will be used in the proof of Theorem~\ref{thm:efcor}  in Section~\ref{sec:ecorloc}.

Let us note that if one defines
\beq \label{eq:hatQ_I}
 \what Q_N^{(L)}(i,j;I) = \sup \left\{ \norm{ Q_{i}\,g(H_N^{(L)})\, Q_{j}}_1 \, \Big| \,
  \mathrm{supp}\, g \subset I, \, |g| \leq 1 \right\},
  \eeq
it is easy to see that $ \what Q_N^{(L)}(i,j;I)\le Q_N^{(L)}(i,j;I)$ (however, with the exception of the case $N=1$, these two quantities are not equal). Thus, for each fixed $N$, Theorem~\ref{thm:efcor} yields exponential decay of $\E( \what Q_N^{(L)}(i,j;I_{1,\delta}))$ in $|i-j|$, uniform in $L$. By known methods available in  the literature on Anderson localization, this decay translates into exponential decay for the corresponding eigenfunction correlators of the infinite volume operators $H_N$. The latter property can then be used to deduce that all $H_N$ (and thus their direct sum, the infinite spin chain Hamiltonian $H$) have pure point spectrum in $I_{1,\delta}$. Since these arguments are fairly standard and the result is only marginally related to the presentation here, we skip a more detailed discussion.

Casting \eqref{eq:efcor5} in the form \eqref{eq:efcor}  allows us to reduce the proof of Theorem~\ref{thm:efcor} to establishing decay properties of the Green's functions associated with the operators $H_N^{(L)}$. In fact, our proofs of these results will also hold for the infinite volume operators $H_N$, which we include because they are of some independent interest.

The Green's function analysis will be done separately along the edge $\mathcal{X}_{N,1}$ (and its finite volume analogs $\mathcal{X}_{N,1}^{(L)} := \mathcal{X}_{N,1} \cap \mathcal{X}_N^{(L)}$) and within  the bulk $\bar{\mathcal{X}}_{N,1} := \mathcal{X}_N \setminus \mathcal{X}_{N,1}$ (and $\bar{\mathcal{X}}_{N,1}^{(L)} := \bar{\mathcal{X}}_{N,1} \cap \mathcal{X}_N^{(L)}$).
In the bulk we have the following (purely deterministic) Combes-Thomas-type bound, which will be proven in Section~\ref{sec:prelim}.

\begin{theorem}[Combes-Thomas bound in the bulk] \label{thm:CTinBulk}
Consider $\pa{\Delta,\lambda}\in   (1,\infty)\times (0,\infty)$,  and let $\bar{H}_{N,1}$ denote the restriction of $H_N$ to $\ell^2(\bar{\mathcal{X}}_{N,1})$. Then there exist constants $C=C(\Delta)<\infty$ and $\eta = \eta(\Delta)>0$, independent of $\lambda$ and $N$, such that
\beq
\| \chi_A (\bar{H}_{N,1}-E-i\epsilon)^{-1} \chi_B\| \le C e^{-\eta \dist_1(A,B)},
\eeq
for all $N\in \N$, $E\in I_{1,\delta}$, $\epsilon\in\R$, and subsets $A$ and $B$ of $\bar{\mathcal{X}}_{N,1}$.

The same bound holds for the restrictions $\bar{H}_{N,1}^{(L)}$ of $H_N^{(L)}$ to $\ell^2(\bar{\mathcal{X}}_{N,1}^{(L)})$, with constants uniform in $N$ and $L$.
\end{theorem}

Here $\dist_1(A,B) = \inf_{x\in A, y\in B} |x-y|$ and $\chi_A$ and $\chi_B$ are the indicator functions of $A$ and $B$.

In fact, we will need and prove results more general than Theorem~\ref{thm:CTinBulk}  in Section~\ref{sec:prelim}, but the above version captures the essence of the type of Combes-Thomas bounds which will be used in this work.

Along the edge we will prove exponential decay of fractional moments of the Green's function in Section~\ref{sec:FMM}, as in the following theorem.

\begin{theorem}[Fractional moment estimate on the edge] \label{thm:fmlocalization} There exists a constant $K>0$ with the following property:  If  the parameters $\pa{\Delta,\lambda}$ are in the region described by  \eqref{eq:paramreg}, there exist constants  $C=C(\Delta)<\infty$ and  $\xi=\xi(\Delta)>0$ (depending only on $\Delta$), such that 
\beq \label{eq:fmlocalization}
\E\left(\abs{\scal{\phi_u,\pa{H_N-E-i\epsilon}^{-1}\phi_v}}^{\frac 12} \right) \le \tfrac {C}{\sqrt{\lambda} }\e^{-\xi\norm{u-v}} ,
\eeq
for all $N\in \N$, $E\in I_{1,\delta}$, $\epsilon>0$, and $u,v\in\mathcal X_{N,1}$.

Moreover, the bound \eqref{eq:fmlocalization} also holds for the operators $H_N^{(L)}$, uniformly in $L$, where $\epsilon=0$ is included.
\end{theorem}

Here $\|u-v\| := \max\{|u_i-v_i|: 1\le i \le N\}$ is the $\infty$-distance.

It is possible to combine Theorems~\ref{thm:CTinBulk} and \ref{thm:fmlocalization} into a {\it global} decay bound for the fractional moments. However, this would essentially require to use the $\infty$-distance (coming from the edge bound), which is insufficient to handle the bulk contributions in the application to eigencorrelators in Theorem~\ref{thm:efcor} (note that in the bulk and for high dimension $N$ the $1$-distance is typically much larger than the $\infty$-distance).  Instead, we will directly combine Theorems~\ref{thm:CTinBulk} and \ref{thm:fmlocalization} (and variants of them) to prove Theorem~\ref{thm:efcor} in Section~\ref{sec:ecorloc}, after providing Wegner-type estimates  and  other technical tools in Sections~\ref{sec:Wegner} and \ref{sec:pairofboxes}.

\subsection{Outline of  Contents}

In Section~\ref{sec:expclustering} we will link Theorem~\ref{thm:efcor}  to our main result on exponential many-body clustering, Theorem~\ref{thm:expclustering}. As noted above, the key idea is to reduce the analysis for general local observables $X$ and $Y$ and general times $t$ to the special case given by\eqref{Ntomany}.

The remaining part of the paper, Sections~\ref{sec:prelim}--\ref{sec:ecorloc}, is devoted to the derivation of the technical results stated in Theorem~\ref{thm:efcor} --  \ref{thm:fmlocalization}, and, in  the  case of Theorems~\ref{thm:CTinBulk} and \ref{thm:fmlocalization}, to some further extensions of these statements used in the proofs.

The key observation behind our arguments there is that for energies in the droplet band $I_{1}$ only the one-quasi-particle sector $\mathcal X_{N,1}$ (the "edge" of $\mathcal X_{N}$) is classically accessible for {\it all} values of $N$. So, a naive removal of the classically forbidden region for each $N$ maps the problem to the well studied  one dimensional  Anderson model (with a correlated random potential). The latter is characterized by complete spectral and dynamical localization.   As it turns out, the underlying localization length is uniformly bounded in $N$, while the density of states in the droplet band rapidly decreases with $N$. These features play an instrumental role in understanding the many body properties of the XXZ model in the sequel.

The rigorous passage from the whole $\mathcal X_{N}$ space to the edge $\mathcal X_{N,1}$ is implemented by means of Schur complementation (known in the physics literature as the Feshbach map). It allows to express the edge-restricted Green's function $G_E$ of the full operator $H_N$ at energy $E$ in terms of the Green's function of an effective Hamiltonian $K_E$ defined on $\mathcal X_{N,1}$. Similar techniques have been used frequently before, explicitly or implicitly, for example in works on random surface potentials (e.g.\ \cite{JL,JM}) or on random potentials restricted to sublattices \cite{ES}.

The operator $K_E$ is comprised of two parts: The first one is simply the restriction of $H_N$ to the edge as in the naive description, while the second part encodes the influence of the bulk.  The technical difficulties associated with the addition of the second term are two-fold: On one hand, it is non-local and non-linear (in $E$), on the other, it is statistically dependent on the randomness associated with the first term. Both issues can have potentially fatal consequences as far as localization is concerned: Non-locality allows for hopping between distant sites while strongly correlated randomness can amplify the effect of resonances, suppressed in the non correlated case.

A significant portion of this paper is devoted to the handling of these two issues. The non-locality of $K_E$ is tackled with  tools developed in Section \ref{sec:prelim}, namely Combes-Thomas bounds for $H_N$ in the bulk, and for a (properly) modified  $H_N$ on the whole $\mathcal X_N$. They allow us to show that $K_E$ is in fact quasi-local for energies in the droplet band, i.e.,  its kernel exhibits rapid spatial decay. Combes-Thomas estimates are widely used  in localization theory; the novel element in our context is the uniform control (in $N$) of the associated decay rate. We then modify the fractional moment proof of localization initially developed by Aizenman and Molchanov in the  one particle context to apply  to the one quasiparticle Hamiltonian $K_E$ in Section \ref{sec:FMM}, resulting in rapid spatial decay for its Green's function in expectation. 

It turns out that Combes-Thomas estimates are also useful in controlling correlations of the random variables; they allow us to effectively decouple Hamiltonians associated with well separated domains in Section \ref{sec:pairofboxes}. A Wegner estimate for $K_E$, which provides a preliminary result pertaining to this decoupling (and is  of some independent value on its own right), is established in Section~\ref{sec:Wegner}, using a strategy from \cite{Stollmann} to handle the non-linearity of $K_E$. Part of the decoupling procedure in Section~\ref{sec:pairofboxes}, specifically Lemma~\ref{lemsepboxes}, can also be seen as a form of the Wegner estimate which is uniform in the correlations within the random potential,   reminiscent of a Wegner estimate  in \cite{KSS}. 

 The main decoupling result of Section~\ref{sec:pairofboxes}, Lemma~\ref{lem:pairb}, follows \cite{ETV} to show that we can extract the conclusions of the energy interval multiscale analysis as in  \cite{DK} from the fractional moment estimate of Section~\ref{sec:FMM}. We then use ideas from proofs of dynamical localization for random Schr\"odinger operators  as in \cite{GKjsp}  to derive
 Theorem~\ref{thm:efcor} in Section~\ref{sec:ecorloc}.   An interesting feature of the argument used there is that small $N$ are treated  using the decoupling between ``boxes" of Lemma~\ref{lem:pairb}, while large $N$ are handled via a large deviation argument reflecting the smallness of the density of states of $H_N$ in $I_1$.

\section{Exponential clustering} \label{sec:expclustering}

We start by showing that our main result, Theorem~\ref{thm:expclustering}, is a consequence of  the $N$-particle eigencorrelator bounds of  Theorem~\ref{thm:efcor}.

We note that $H^{(L)}$ almost surely has simple spectrum. A simple analyticity based argument for this can be found in Appendix A of \cite{ARS} (the argument is presented there for the XY chain, but it holds for every random operator of the form $H_0 + \sum_{k=-L}^L \omega_k \mathcal{N}_k$ in $\otimes_{k=-L}^L \C^2$, as in our case). Thus, almost surely, all its normalized eigenstates can be labeled as $\psi_E$ with the corresponding eigenvalue $E$ of $H^{(L)}$. Moreover, $\psi_E$ belongs to one of the $N$-particle sectors of $H^{(L)}$. 

Given a local observable $X$ with support $\cS_X$, we define projections $P_{\pm}^{\pa{X}}$ by 
\beq
P_+^{\pa{X}}= \bigotimes_{j\in \cS_X}\ (1-\cN_i) \qtx{and} P_-^{\pa{X}}=1-P_+^{\pa{X}}
\eeq 
Note that $P_\pm^{\pa{X}}$ are supported on $\cS_X$, $P_+^{\pa{X}}$ projects onto the basis states $e_{\alpha}$ (see \eqref{eq:prodbasis}) with no particles in $\cS_X$, i.e.\ $\alpha_j=0$ for all $j\in \cS_X$, and $P_-^{\pa{X}}$ projects onto states with at least one particle in $\cS_X$. In particular, we have
\beq\label{PsuppX}
 P_-^{\pa{X}}\le \sum_{i\in\cS_X}  \cN_i .
\eeq 
In addition,  $[P_{\pm}^{\pa{X}},\cN]=0$ (where   $\cN=\cN^{(L)}$)  and 
\beq\label{Xdecomp}
 X =\sum_{a,b \in \set{+,-}}X^{a,b}, \qtx{where} X^{a,b}= P_{a}^{\pa{X}}XP_{b}^{\pa{X}},
 \eeq
 all of which are supported on $\cS_X$.
Moreover,  since $P_+^{\pa{X}}$ is a rank one projection on $\cS_X$, we must  have 
\beq\label{Xzeta}
X^{+,+}=\zeta P_+^{\pa{X}}, \qtx{where} \zeta\in \C, \ \abs{\zeta} \le \|X\|.
\eeq

\begin{lemma}\label{lem:clusgen}  Let $I\subset \R$ be an interval, set $\sigma_I= \sigma(H^{(L)}) \cap I$, and assume that all eigenvalues $E$ in $\sigma_I$ are simple.
Then for 
any two local observables $X$ and $Y$ with $\cS_X \cap \cS_Y = \emptyset$,   we have (see \eq{defcRXY})
\begin{align}\label{estcorrel}
\sup_{t\in \R}   \cR_{\tau^I_t\pa{X},Y}(I)\le C \|X\| \|Y\|\sum_{E\in\sigma_{I}}\norm{ P_-^{\pa{X}} \psi_E}\norm{ P_-^{\pa{Y}}\psi_E}.
\end{align} 
\end{lemma}   

\begin{proof} 
Let $F\subset I$.
It follows from \eq{eq:correlfulK} and \eq{Xdecomp} that 
\beq\label{Rdecomp}
R_{\tau^I_t\pa{X},Y}(F)\le \sum_{a,b,c,d\in \set{+,-}} R_{\tau^I_t\pa{X^{a,b}},Y^{c,d}}(F).
\eeq
Let  $\psi$ be a many body state with $\norm{\psi}=1$ and $\cN \psi =N \psi$, where $N\in \N$.
Then  for  $Z^{a,b}=\tau^I_t\pa{X^{a,b}}$ or $Z^{a,b}=Y^{a,b}$, we have (using, in particular, that  $e^{\pm i tH_I^{(L)}}$ leaves the particle sectors invariant)
\beq\label{ZNupdown}
\chi_{ [0,N-1]  } (\cN)  Z^{+,-}\psi=Z^{+,-}\psi \qtx{and} \chi_{ [N+1,\infty)  } (\cN)  Z^{-,+}\psi=Z^{-,+}\psi .
\eeq
Since it follows from \eq{eq:correlfulK} that (we mostly omit $L$ from the notation)
\beq\label{eq:correlfulK2} 
R_{\tau^I_t\pa{X},Y}(F)=\abs{\sum_{E\in \sigma_F}     \scal{\psi_E, \tau^I_t\pa{X} \bar P_F Y \psi_E }},
\eeq
 where $\sigma_F= \sigma(H^{(L)}) \cap F$,  we conclude that 
\beq \label{trick1}
 R_{\tau^I_t\pa{X^{+,-}},Y^{+,-}}(F)=R_{\tau^I_t\pa{X^{-,+}},Y^{-,+}}(F)=0.
\eeq
In addition, it follows from  \eq{eq:correlfulK},  \eq{Xzeta} and $P_+^{\pa{X}}=1-P_-^{\pa{X}}$ that for any observable $Z$ we have
\beq \label{trick2}
R_{\tau^I_t\pa{X^{+,+}},Z}(F)= \abs{\zeta} R_{\tau^I_t\pa{P_-^{\pa{X}}},Z}(F), \qtx{where} \zeta\in \C, \ \abs{\zeta} \le \|X\|.
\eeq
Furthermore, for any two local observables $Z$ and $W$  we have 
\beq \label{trick3}
R_{\tau^I_t\pa{Z},W}(F)=R_{\tau^I_{-t}\pa{W^*},Z^*}(F). 
\eeq

Thus,  using \eqref{trick1} -- \eqref{trick3}, one can reduce estimating all sixteen terms  in \eqref{Rdecomp}  to estimating  correlators of the form 
\begin{align} \label{remcorr}
 & R_{{\tau^I_t\pa{X^{-,-}}},Y^{-,-}}(F), \quad R_{\tau^I_t\pa{X^{-,+}},Y^{+,-}}(F), \quad R_{{\tau^I_t\pa{X^{-,-}}},Y^{+,-}}(F), \\ &  R_{\tau^I_t\pa{X^{+,-}},Y^{-,-}}(F) \qtx{and} R_{{\tau^I_t\pa{X^{+,-}}},Y^{-,+}}(F),\label{remcorr2}
 \end{align} 
where $X$ and $Y$ are local observables  with disjoint supports with $X^{+,+} =Y^{+,+}=0$.

 The three cases  in \eq{remcorr} can be treated together.
If $a,b\in\set{+,-}$, we have
\begin{align}\notag
 R_{\tau^I_t\pa{X^{-,b}},Y^{a,-}}(F) &\le 
\sum_{E\in \sigma_{F} }\abs{\scal{\psi_E, \tau^I_t\pa{X^{-,b}}\bar P_F Y^{a,-}\psi_E}}\\  & \le  \|X\| \|Y\|\sum_{E\in\sigma_{F}}\norm{ P_-^{\pa{X}} \psi_{E}} \norm{P_-^{\pa{Y}}\psi_{E}} .   \label{scalabscal}
 \end{align}
 
 It remains to consider  the two  cases in \eq{remcorr2}, which are of the form $R_{\tau^I_t\pa{X^{a,-}},Y^{-,b}}(F)$ with  $a,b\in\set{+,-}$, and can be treated together.
We have
 \begin{align}\label{+--+777}
  R_{\tau^I_t\pa{X^{a,-}},Y^{-,b}}(F) \le \abs{\tr \pa{P_F \tau^I_t\pa{X^{a,-}}P_I \bar P_F Y^{-,b} P_F}}
  +  \abs{\tr \pa{P_F \tau^I_t\pa{X^{a,-}}\bar P_I  Y^{-,b} P_F}}, 
 \end{align}
  where we used $\bar P_I \bar P_F= \bar P_I$.

 We estimate the first term by
\begin{align}\notag
\abs{\tr \pa{P_F \tau^I_t\pa{X^{a,-}}P_I \bar P_F Y^{-,b} P_F}}&\le \sum_{E\in \sigma_{F}} \abs{\scal{\psi_E, X^{a,-}\e^{-itH_I}P_I \bar P_F Y^{-,b}\psi_E}}\\ &  \le \sum_{E\in \sigma_{F}} \sum_{E'\in \sigma_{I}} \abs{\scal{\psi_E, X^{a,-}\psi_{E'}}}\abs{ \scal{\psi_{E'},Y^{-,b}\psi_E}}.
\label{firstterm}
\end{align} 

 For the second term, we use
 \begin{align}\notag
&\abs{\tr \pa{P_F \tau^I_t\pa{X^{a,-}}\bar P_I   Y^{-,b} P_F}}=  \abs{\tr \pa{\e^{itH_I} P_F X^{a,-}\bar P_I Y^{-,b}P_F}}\\ & \quad 
 \le \abs{\tr \pa{\e^{itH_I} P_F X^{a,-}P_+^{\pa{Y}} \bar P_I Y^{-,b}P_F}} + \abs{\tr \pa{\e^{itH_I} P_F X^{a,-}P_-^{\pa{Y}} \bar P_I Y^{-,b}P_F}}.
 \label{lastterm1}  \end{align}
We have   $X^{a,-}P_+^{\pa{Y}}  \bar P_I Y^{-,b}=-P_+^{\pa{Y}}X^{a,-}\ P_I Y^{-,b}$, since $P_+^{\pa{Y}}P_-^{\pa{Y}}=0$ and $X^{a,-}P_+^{\pa{Y}}= P_+^{\pa{Y}}X^{a,-}$ (due to their disjoint supports), so we get
\begin{align}\notag
& \abs{\tr \pa{\e^{itH_I} P_F X^{a,-}P_+^{\pa{Y}} \bar P_I Y^{-,b}P_F}}=\abs{\tr \pa{\e^{itH_I} P_F P_+^{\pa{Y}} X^{a,-} P_I Y^{-,b}P_F}} 
\\ & \qquad \qquad \le \sum_{E\in \sigma_{F}} \sum_{E'\in \sigma_{I}} \abs{ \scal{\psi_{E},P_+^{\pa{Y}} X^{a,-} \psi_{E'}}}
\abs{\scal{\psi_{E'},  Y^{-,b}\psi_{E}}}.
\end{align}

To estimate the second term in \eqref{lastterm1}, we observe that
\begin{align}\notag
&\abs{\tr \pa{\e^{itH_I} P_F X^{a,-}P_-^{\pa{Y}}   \bar P_IY^{-,b} P_F}} 
\\ & \quad \le \abs{\tr \pa{\e^{itH_I} P_F P_-^{\pa{Y}}X^{a,-}  \bar P_I P_+^{\pa{X}}Y^{-,b} P_F}} + \abs{\tr \pa{\e^{itH_I} P_F P_-^{\pa{Y}}X^{a,-}   \bar P_IY^{-,b}P_-^{\pa{X}} P_F}}.\label{+--+666}
\end{align}
We have 
\begin{align}\notag
 & \abs{\tr \pa{\e^{itH_I} P_F P_-^{\pa{Y}}X^{a,-}  \bar P_I P_+^{\pa{X}}Y^{-,b} P_F}}  = \abs{\tr \pa{\e^{itH_I} P_F P_-^{\pa{Y}}X^{a,-}   P_I Y^{-,b}P_+^{\pa{X}} P_F}} \notag  \\ & \qquad \qquad  \le  \sum_{E\in \sigma_{F}} \sum_{E'\in \sigma_{I}}  \abs{\scal{\psi_E, P_-^{\pa{Y}} X^{a,-}\psi_{E'}}}\abs{ \scal{\psi_{E'},Y^{-,b}P_+^{\pa{X}}\psi_E}}
,\label{+--+7778}
\end{align}
where we have again exploited the fact that $X$ and $Y$ have disjoint supports
and $ X^{a,-}  \bar P_I P_+^{\pa{X}}Y^{-,b}=-X^{a,-}   P_I Y^{-,b}P_+^{\pa{X}}$. In addition,
\begin{align}\label{+--+7578}
\abs{\tr \pa{\e^{itH_I} P_F P_-^{\pa{Y}}X^{a,-}   \bar P_IY^{-,b}P_-^{\pa{X}} P_F}} \le \|X\| \|Y\|\sum_{E\in\sigma_{F}}\norm{ P_-^{\pa{Y}} \psi_{E}} \norm{P_-^{\pa{X}}\psi_{E}}.
\end{align}

 Now  let $\cF=\set{F_q}_{q\in \cQ}\in \cF_I$.  Then for any operators $Z_1,Z_2$ we have
 \begin{align}\notag
 &\sum  _{q\in \cQ}\sum_{E\in \sigma_{F_q}} \sum_{E'\in \sigma_{I}} \abs{\scal{\psi_E, Z_1\psi_{E'}}}\abs{ \scal{\psi_{E'},Z_2\psi_E}}\\ & \;  \le \sum_{E'\in \sigma_{I}} \pa{\sum_{E\in \sigma_{I}} \abs{\scal{\psi_E, Z_1\psi_{E'}}}^2}^{1/2} \pa{\sum_{E\in \sigma_{I}} \abs{ \scal{\psi_{E'},Z_2\psi_E}}^2}^{1/2} \le \sum_{E'\in \sigma_{I}} \norm{Z_1\psi_{E'}} \norm{Z_2^*\psi_{E'}},
\end{align}
where we have used the Cauchy Schwarz   and Bessel's inequalities. It thus  follows from \eq{+--+777}-\eq{+--+7578} that 
\begin{align}\label{+--+98}
 R_{\tau^I_t\pa{X^{a,-}},Y^{-,b}}(\cF) =\sum_{q\in \cQ} R_{\tau^I_t\pa{X^{a,-}},Y^{-,b}}(F_q) \le  4\|X\| \|Y\| \sum_{E\in \sigma_{I}} \norm{ P_-^{\pa{X}} \psi_{E}} \norm{P_-^{\pa{Y}}\psi_{E}}.
\end{align} 

We conclude from 
 \eq{scalabscal} and \eq{+--+98} that, for local observables $X,Y$ with $X^{+,+} =Y^{+,+}=0$,   \begin{align}\label{estcorrel333}
\sup_{t\in \R}   R_{\tau^I_t\pa{X},Y}(\cF)\le  C \|X\| \|Y\| \sum_{E\in \sigma_{I}} \norm{ P_-^{\pa{X}} \psi_{E}} \norm{P_-^{\pa{Y}}\psi_{E}}.
\end{align}  
 \end{proof}  

We are ready to prove Theorem~\ref{thm:expclustering}.
        
\begin{proof}[Proof of Theorem~\ref{thm:expclustering}] Let $I=  I_{1,\delta}$.
Given two local observables $X$ and $Y$ with $\max \cS_X < \min \cS_Y$ (or vice-versa), it follows from Lemma~\ref{lem:clusgen} that, almost surely, 
\begin{gather} \label{eq:expclustering54}
 \sup_{t\in \R} \cR_{\tau^{I}_t\pa{X},Y}(I) \le C \|X\| \|Y\|\sum_{E\in \sigma_{ I}}\norm{ P_-^{\pa{X}} \psi_E}\norm{ P_-^{\pa{Y}}\psi_E}.
   \end{gather}
  
In view of \eq{PsuppX}, we have 
\beq
\norm{ P_-^{\pa{Z}} \psi_E}\le \sum_{i\in \cS_Z} \norm{ \cN_i \psi_E}\qtx{for all} E\in \sigma_{ I}, \qtx{where} Z=X,Y.
\eeq
Thus,  if  the parameters $\pa{\Delta,\lambda}$ are in the region described by  \eqref{eq:paramreg} as in Theorem~\ref{thm:efcor},    it follows from this  theorem that
\begin{align}
\E\pa{\sup_{t\in \R}  \cR_{\tau^{I}_t\pa{X},Y}(I)}  
\le  C \norm{X} \norm{Y} \sum_{i\in \cS_ X, \,  j\in \cS_ Y} e^{-m|i-j|}.
\end{align}
This estimate  implies   \eq{eq:expclustering0}, since  given $J, K \subset \Z$ with $\max J < \min K$ we have
 \begin{align}\notag
 \sum_{j\in J,k\in K} \e^{-m\abs{j-k}}&\le  \sum_{\substack {j \in \Z , j\le 0\\k \in \Z, k\ge  \dist (J,K)}}\e^{-m(k-j)}= \e^{-m\dist (J,K)}\sum_{\substack {j \in \Z , j\le 0\\k \in \Z, k\ge  0}}\e^{-m (k-j)}
 \\  & =\e^{-m\dist (J,K)}\pa{\sum_{k=0}^\infty\e^{-m k}}^2= \pa{1-\e^{-m}}^{-2}\e^{-m\dist (J,K)}.
 \end{align}
Theorem~\ref{thm:expclustering}  is proven.
\end{proof}

\section{Combes-Thomas bounds}\label{sec:prelim}


A key feature to be exploited in our analysis is that the ``local dimension'' of the graph $\mathcal{X}_N$, when measured by the number of next neighbors of a vertex, gradually increases from being one-dimensional at the edge $\mathcal{X}_{N,1}$ to $N$-dimensional in the deep bulk (where $\mathcal{X}_N$ locally looks like $\Z^N$). This can be made precise with the following definitions.

Let
\beq
\mathcal X_{N,k} = \set{x\in \mathcal X_N:\ \tilde W(x) \le k}, \quad \bar{\mathcal X}_{N,k} = \mathcal{X}_N \setminus \mathcal{X}_{N,k},
\eeq
with corresponding decomposition of $\mathcal H$ into subspaces ${\mathcal H}_k=\ell^2(\mathcal{X}_{N,k})$ and $\bar{\mathcal H}_k=\ell^2(\bar{\mathcal X}_{N,k})$. Let $P_k$ and $\bar{P}_k = I - P_k$ be the orthogonal projections onto ${\mathcal H}_k$ and $\bar{\mathcal H}_k$, respectively. Equivalently, $P_k$ is the projection onto states $\phi_x$ where the configuration $x=(x_1,\ldots,x_N)$ consists of at most $k$ ``clusters'' (i.e.\ connected components of nearest neighbors) and $\bar P_k$ projects onto states with at least $k+1$ clusters. For an operator $A$ on $\ell^2(\mathcal X_N)$, we will denote by $A_k = P_k A P_k$ and $\bar A_k = \bar P_k A\bar P_k$ the restrictions of $A$ to ${\mathcal H}_k$ and $\bar{\mathcal H}_k$. 

For $k=1$ this gives the one-dimensional edge $\mathcal{X}_{N,1}$ defined in \eqref{edge} above. Note that a vertex $x=(x_1, x_1+1, \ldots, x_1+N-1)$ in $\mathcal{X}_{N,1}$ has only two nearest neighbors in $\mathcal{X}_N$, $(x_1-1, x_1+1,\ldots,x_1+N-1)$ and $(x_1,\ldots, x_1+N-2, x_1+N)$. Similarly, the vertices $x$ in $\mathcal{X}_{N,k} \setminus \mathcal{X}_{N,k-1}$, consisting of $k$ clusters, have exactly $\tilde{W}(x) = 2k$ nearest neighbors. Thus the transitional region of $\mathcal{X}_N$ between its edge and deep bulk has geometric structure highly uniform in $N$. As this region determines the bottom of the spectrum of $H$ in the Ising phase, we will be able to prove operator bounds with constants uniform in $N$.

Most prominently, this geometric structure of the graphs $\mathcal{X}_N$ will be exploited in the following Combes-Thomas bound suitable for our model. This result holds not only in the (single) droplet spectrum $I_{1,\delta}$, but for the larger ``$k$-droplet intervals''
\beq
I_{k,\delta} := [1-\tfrac{1}{\Delta}, (k+1-\delta)(1-\tfrac{1}{\Delta})].
\eeq

The restrictions $\overline{\pa{H_N}}_k$ of $H_N$ to $\bar{\mathcal{H}}_k$ satisfy the lower bound 
\beq \label{eq:klow}
\overline{\pa{H_N}}_k \ge (k+1)\pa{1-\tfrac{1}{\Delta}}. 
\eeq
Indeed, as the negative Laplacian and the random potential are non-negative, one has the lower bound $H_N \ge \pa{1-\frac{1}{\Delta}}\tilde W$, which yields (\ref{eq:klow}) after restricting to $\bar{\mathcal{H}}_k$. Thus Combes-Thomas bounds for $\overline{\pa{H_N}}_k$ to energies $E\in I_{k,\delta}$ are a standard fact. However, the constants in these bounds generally depend on the dimension of the graph, compare, e.g.,   \cite[Section~11.2]{Kirsch}, \cite[Appendix~B]{KNR}. That we can prove an $N$-independent bound here is due to the special structure of  the graphs $\mathcal{X}_N$.

It will be convenient for the latter purposes to formulate and prove this result in a greater generality. To this end, we need to recollect a few facts from graph theory:  A {\it graph} is an ordered pair $G=(\mathcal{V}, \mathcal{E})$  comprising a set $\mathcal{V}$ of vertices together with a set $\mathcal{E}$ of edges. A {\it subgraph} of a graph $G$ is another graph formed from a subset of the vertices and edges of $G$. The vertex subset must include all endpoints of the edge subset, but may also include additional vertices. An {\it induced subgraph} of a graph is a subgraph formed from a subset of the vertices of the graph and {\it all} of the edges connecting pairs of vertices in that subset. For a graph $G=(\mathcal{V}, \mathcal{E})$ we define its Laplacian as an operator on $\ell^2\pa{\mathcal{V}}$, acting as 
\beq\label{eq:graphlapl} 
\pa{\mathcal{L}^{(G)}\psi}(x)=\sum_{\substack{y\in\mathcal{V}\\ \scal{y,x}\in  \mathcal{E}}} (\psi(y)-\psi(x)),\quad \psi\in \ell^2\pa{\mathcal{V}}.
\eeq
The operator $-\mathcal{L}^{(G)}$ is positive definite.  For an operator $H$ on 
a graph $G=(\mathcal{V}, \mathcal{E})$ and $\mathcal{V'}\subset V$, we will denote by $H_{\mathcal{V'}}$ the restriction of  $H$ to $\mathcal{V'}$, i.e. $H_{\mathcal{V'}}=P_{\mathcal{V'}}HP_{\mathcal{V'}}$, where $P_{\mathcal{V'}}$ is the indicator of the set $\mathcal{V'}$.  If $V$ is a potential, we will usually simply write $V$ for $V_\cV$. In particular, we usually write $P_k$ and $\bar P_k$  for $(P_k)_\cV$ and $(\bar P_k)_\cV$.

If $G'=(\mathcal{V'}, \mathcal{E'})$ is an induced subgraph for $G=(\mathcal{V}, \mathcal{E})$, then the restriction $\pa{\mathcal{L}^{(G)}}_{\mathcal{V'}}$ of $\mathcal{L}^{(G)}$ to $G'$ satisfies
\beq
-\pa{\mathcal{L}^{(G)}}_{\mathcal{V'}}\ge -\mathcal{L}^{(G')}.
\eeq
An immediate consequence of this relation is \eqref{eq:klow} (with $G=\mathcal{X}_N$ and $G'=\bar{\mathcal{X}}_{N,k}$).

Given subsets   $A$ and $B$ of $\cX_N$, we set $\dist_1\pa{A,B}:=\inf_{x\in A,y\in B}\abs{x-y}$.

\begin{proposition}[Combes-Thomas bound]\label{lem:CT}
Let $G=(\mathcal{V}, \mathcal{E})$ be a subgraph of $\mathcal X_N$, let $V$ be a bounded potential on $\mathcal X_N$, consider the self-adjoint operator 
\beq\label{defHG}
H_N^{\pa{G}}=-\frac{1}{2\Delta}\mathcal{L}^{(G)}
+\pa{1-\tfrac{1}{\Delta}}\tilde W_{\mathcal{V}}+V_{\mathcal{V}} \qtx{on} \ell_2\pa{\mathcal{V}}.
\eeq
Let $z \notin \sigma(H_N^{\pa{G}})$ with
\beq\label{etaz}
 \norm{\tilde W^{\frac 12} (H_N^{\pa{G}}-z)^{-1}\tilde W^{\frac 12}}\le \tfrac 1 {\kappa_z} < \infty
\eeq
Then for  all  subsets $A$ and $B$ of $\mathcal{V}$ we have
\begin{align}\label{2a}
\| \chi_{A} \pa{H_N^{\pa{G}}-z}^{-1}\chi_{B}\|\le 
\| \chi_{A} \tilde W^{\frac 12}\pa{H_N^{\pa{G}}-z}^{-1}\tilde W^{\frac 12} \chi_{B}\| \le  \tfrac 2 {\kappa_z}\, e^{-\eta_z \dist_1\pa{A, B}},
\end{align}
where $\eta_z=  \log \pa{1 + \frac { \kappa_z\Delta}{2 }}$.
\end{proposition}

This proposition  implies  the following Combes-Thomas bounds for $H_N$ in the bulk.
We formulate our results in enough generality so they apply not only  to the operator $H_N$ on  $\ell^2(\mathcal{X}_N)$ given in \eq{HNinfty} and  to the finite volume operators $H_N^{(L)}$    on $\ell^2(\mathcal{X}_N^{(L)})$ defined in \eq{eq:H_N}, but also to  similarly defined operators on subgraphs of $\mathcal{X}_N$ and $\mathcal{X}_N^{(L)}$, and their restrictions.    For the Combes-Thomas bounds, this can be accomplished by proving bounds on subgraphs which are uniform with respect to the addition of a nonnegative potential, as in the following corollary.  Note that the constants in the Combes-Thomas bounds are independent of  the disorder parameter $\lambda$.

\begin{corollary}[Combes-Thomas bound in the bulk]\label{lem1}Let $G=(\mathcal{V}, \mathcal{E})$ be a subgraph of $\mathcal X_N$,  and consider the operator
\beq\label{HNG}
H_N^{\pa{G}}=-\tfrac{1}{2\Delta}\mathcal{L}^{(G)}
+\pa{1-\tfrac{1}{\Delta}}\tilde W_{\mathcal{V}}+ \lambda (V_\omega)_{\mathcal{V}}  + Y \qtx{on}
\ell_2\pa{\mathcal{V}},
\eeq
where $Y$ is a nonnegative bounded potential on $\cV$. Then for all $k \in \N$,  
 $E\in I_{k,\delta}$,  $\epsilon \in \R$, and subsets $A$ and $B$ of $\bar{\mathcal{X}}_{N,k}\cap \cV$, we have
\begin{align}\notag
\| \chi_{A} \overline{\pa{H_N^{\pa{G}}-E-i\epsilon}}_k^{-1}\ \chi_{B}\|& \le  \tfrac 1{k+1}
\| \chi_{A} \tilde W^{\frac 12}\overline{\pa{H_N^{\pa{G}}-E-i\epsilon}}_k^{-1}\tilde W^{\frac 12} \chi_{B}\| \\ &  \le  C e^{-\eta \dist_1\pa{A, B}}\label{2},
\end{align}
with the inverse taken on $\bar{\mathcal H}_k\cap \ell_2\pa{\mathcal{V}}$, where
\beq\label{Ceta}
C=C(\Delta,\delta)=\tfrac {4\Delta} {\delta(\Delta -1)}\qtx{and} \eta = \eta(\Delta,k,\delta)=  \log \pa{1 + \tfrac {\delta (\Delta-1)}{4(k+1) }}.
\eeq
In addition,  for all $E\in I_{1,\delta}$,  $\epsilon \in \R$, and   subsets $A$ and $B$ of $\cV$,  
\begin{align} \notag
\| \chi_A(H^{\pa{G}}_N+P_1-E-i\epsilon)^{-1} \chi_B \| &\le
\| \chi_A \tilde W^{\frac 12}(H^{\pa{G}}_N+P_1-E-i\epsilon)^{-1} \tilde W^{\frac 12}\chi_B \| \\ 
& \le  C^\pr e^{-\eta^\pr \dist_1(A,B)},\label{eq:fullCT}
\end{align}
where
\beq\label{eq:fullCT2}
C^\pr=C^\pr(\Delta,\delta)=\tfrac {8\Delta} {\delta(\Delta -1)}\qtx{and} \eta^\pr = \eta^\pr(\Delta,\delta)=  \log \pa{1 + \tfrac {\delta (\Delta-1)}{8 }}.
\eeq
\end{corollary}

\begin{proof}[Proof of Proposition \ref{lem:CT}]
We will write $H := H^{\pa{G}}$ and $R_z := (H-z)^{-1}$ and use notation $\tilde{W}$ also for its restriction $\tilde W_{\mathcal{V}}$ to $\mathcal{V}$.  
Consider a distance function $\rho_A$ on $\mathcal{X}_{N}$ defined as $\rho_A(x)=\dist_1\pa{A,x}$. For $\eta\in \R$, let ${H}_\eta$ be a dilation of ${H}$ with respect to $\rho_A$, i.e., ${H}_\eta =\e^{-\eta \rho_A}\, {H}\, \e^{\eta \rho_A}$,
and set $K_\eta:=\e^{-\eta \rho_A}\, {H}\, \e^{\eta \rho_A}-{H}$. A straightforward computation yields 
\beq\label{eq:diffdilat}
\pa{K_\eta\psi}(x)=\tfrac 1 {2 \Delta}\sum_{\substack{y\in \mathcal{V}:\\ \scal{x,y}\in \mathcal{E}}}\pa{\e^{\eta \pa{\rho_A(y)-\rho_A(x)}}-1}\psi(y)
\eeq
for $\psi\in \ell^2\pa{\mathcal{V}}$.  If we denote $f_\eta(x,y):=\e^{\eta \pa{\rho_A(y)-\rho_A(x)}}-1$, 
then it follows from the definition of $\rho_A$ that $\abs{f_\eta(x,y)}\le \e^{\eta}-1\mbox{ for } \abs{x-y}=1$. We note that for any $x\in \mathcal X_N$,
\beq\label{eq:W'}\sum_{\substack{w\in \mathcal X_N\\ |x-w|=1}}1=2\tilde  W(x), \eeq
so  for a subgraph $G$ of $\mathcal X_N$ we have
\beq\label{eq:W's}\sum_{\substack{w \in \mathcal{V}:\\ \scal{x,w}\in \mathcal{E}}}1\le 2\tilde  W(x). \eeq

Thus \eqref{eq:diffdilat} and \eqref{eq:W'} imply
\begin{align} \notag
&\scal{ \tilde W^{-1/2}\,K_\eta\,\tilde W^{-1/2} \psi,  \tilde W^{-1/2}\,K_\eta\,\tilde W^{-1/2}\psi }  
\\ & \qquad \notag = \tfrac 1 {4 \Delta^2}
\sum_{\substack{x, y, w \in \mathcal{V}:\\ \scal{x,y}\in \mathcal{E}, \scal{x,w}\in \mathcal{E}}}\tilde W^{-1}(x)\tilde W^{-1/2}(y)\tilde W^{-1/2}(w)f_\eta(x,w)f_\eta(x,y)\psi(y)\bar\psi(w)  \\ 
  & \qquad \notag\le   \tfrac 1 {4 \Delta^2}
\pa{\e^{\eta}-1}^2\sum_{\substack{x, y, w \in \mathcal{V}:\\ \scal{x,y}\in \mathcal{E}, \scal{x,w}\in \mathcal{E}}}\tilde W^{-1}(x)\tilde W^{-1}(y)\abs{\psi(y)}^2  \\  & \qquad 
 \le   \tfrac 1 {2 \Delta^2}
 \pa{\e^{\eta}-1}^2 \sum_{\substack{x, y \in \mathcal{V}:\\ \scal{x,y}\in \mathcal{E}}} \tilde{W}^{-1}(y)|\psi(y)|^2 \le  \tfrac 1 { \Delta^2} \pa{\e^{\eta}-1}^2 \|\psi\|^2, \label{longcalc}
\end{align}
where the final line uses \eqref{eq:W's} and that $\tilde{W}^{-1}(y) \le 2 \tilde{W}^{-1}(x)$ for adjacent sites $x$ and $y$. This implies the bound
\beq\label{eq:relb'}
\left\|\tilde W^{-1/2}\,K_\eta\,\tilde W^{-1/2}\right\|\le  \tfrac 1 {\Delta}\pa{\e^{\eta}-1}.
\eeq
If $\kappa_z$ is as  in  \eq{etaz}, and 
\beq\label{tcondCT}
  \tfrac 1 {\Delta}  \pa{\e^{\eta}-1} \kappa_z^{-1} \le \tfrac 12 , \qtx{i.e.,} \eta \le \log \pa{1 + \tfrac {\kappa_z \Delta}{2 } } ,
\eeq  
 we have
\begin{align}
&\tW^{\frac 12} (H_\eta-z)^{-1}\tW^{\frac 12} \notag \\ 
& \hskip30pt = \pa{\tW^{\frac 12} (H-z)^{-1}\tW^{\frac 12}}\sum_{j=0}^\infty
\pa{ \pa{\tW^{-\frac 12} K_\eta \tW^{-\frac 12}}\pa{\tW^{\frac 12} (H-z)^{-1}\tW^{\frac 12} }}^j,
\end{align}
so
\beq
\norm{\tW^{\frac 12} (H_\eta-z)^{-1}\tW^{\frac 12}}\le \tfrac 1{2 \kappa_z}.
\eeq

Now let $A, B\subset \cV$ and $\eta = \log \pa{1 + \tfrac {\kappa_z \Delta}{2 } }$.  We have
 \begin{align}
& \norm{\chi_{A}\tW^{\frac 12} (H-z)^{-1}\tW^{\frac 12} \chi_B}=\norm{\chi_{A}\e^{\eta \rho_A}\tW^{\frac 12} (H_{\eta}-z)^{-1}\tW^{\frac 12}\e^{\eta}\e^{-\eta \rho_A} \chi_B} \notag \\ & \quad \le \norm{\tW^{\frac 12} (H_{\eta}-z)^{-1}\tW^{\frac 12}}
\norm{\e^{-\eta \rho_A} \chi_\Psi}\le  \tfrac 1{2 \kappa_z}\e^{-\eta \dist_1(A,B)}.
 \end{align}
 \end{proof}
 
 To prove Corollary~\ref{lem1} we will use the following lemma.  If $\cV \subset \cX_N$, we set
  $\cV_k =\cV  \cap \cX_{N,k}$ and $\bar \cV_k =\cV \setminus\cV_k = \cV  \cap \bar \cX_{N,k}$

\begin{lemma} \label{lemHY} Let $G=(\mathcal{V}, \mathcal{E})$ be a subgraph of $\mathcal X_N$, $Y$ a nonnegative potential on $\cV$, and $H_N^{\pa{G}}$ as in \eq{HNG}.    Then, for  $k\in \set{1,2,\ldots,N-1}$, we have
\beq\label{barHk}
\bar{H}^{(G)}_{k}\ge (k+1)\pa{1-\tfrac{1}{\Delta}} \qtx{on} \ell_2(\bar {\cV}_k).
\eeq
Moreover, if $E \in I_{k,\delta}$, then $E\notin \sigma (\bar{H}^{(G)}_{k})$ and for all $\eps \in \R$ we have
\begin{align}\label{tWbarHtWeps}
\norm{\tilde{W}^{1/2} \pa{\bar H^{(G)}_k-E-i\eps}^{-1}\tilde{W}^{1/2} }\le \tfrac {2(k+1)} {\delta \pa{1-\tfrac{1}{\Delta}}}.
\end{align}
In the special case  $\eps=0$ we have
\begin{align}\label{tWbarHtW}
\norm{\tilde{W}^{1/2} \pa{\bar{H}^{(G)}_{k}-E}^{-1} \tilde{W}^{1/2}} &\le \tfrac {k+1} {\delta \pa{1-\tfrac{1}{\Delta}}}.
\end{align}

\end{lemma}

\begin{proof} It follows from \eq{HNG} that
\beq
H_N^{(G)}\ge  \pa{1-\tfrac{1}{\Delta}}\tilde W \qtx{on} \ell_2\pa{\mathcal{V}},
\eeq
so
\beq\label{barHk22}
\bar H_k^{(G)}\ge  \pa{1-\tfrac{1}{\Delta}}\tilde W\ge (k+1) \pa{1-\tfrac{1}{\Delta}} \qtx{on} \ell_2\pa{\mathcal{V}_k},
\eeq
and \eq{barHk} follows by restriction to $\ell^2(\bar {\cV}_k)$.

If $E \in I_{k,\delta}$, we have
\begin{align}\notag
\bar H_k^{(G)}-E &\ge \pa{1-\tfrac{1}{\Delta}}\tilde W- \tfrac E {k+1}  \tW \ge 
  \pa{1-\tfrac{1}{\Delta}}\tilde W -  (1- \tfrac \delta {k+1} )(1-\tfrac{1}{\Delta})
\tW \\  & \ge \tfrac \delta {k+1}(1-\tfrac{1}{\Delta}) \tW
\ge \delta (1-\tfrac{1}{\Delta})>0, \label{barHE}
\end{align}
so $E\notin \sigma (\bar H_k^{(G)})$ and 
\beq
\tilde{W}^{-1/2} \pa{\bar H_k^{(G)}-E} \tilde{W}^{-1/2}\ge \tfrac \delta {k+1}(1-\tfrac{1}{\Delta}),
\eeq
which implies \eq{tWbarHtW}.

If  $\eps \ne 0$,   we have, using \eq{barHE},
\begin{align}
0&\le \Rea \pa{\bar  H_k^{(G)}-E-i\eps}^{-1}\le \pa{\bar  H_k^{(G)}-E}^{-1}, \notag
 \\ 0&\le \tfrac {\eps}{\abs{\eps}}\Ima \pa{\bar  H_k^{(G)}-E-i\eps}^{-1}\le \pa{\bar  H_k^{(G)}-E}^{-1} .
\end{align}
Thus
\eq{tWbarHtWeps} follows immediately  from \eq{tWbarHtW} and the triangle inequality  for norms.
\end{proof}

\begin{proof}[Proof of Corollary~\ref{lem1}]
The estimate \eq{2} follows immediately from \eq{2a} and \eq{tWbarHtWeps}.

To prove \eq{eq:fullCT}, note $P_1$ is a bounded potential. If $E \in I_{1,\delta}$, we have
\begin{align}
P_1\pa{\pa{1-\tfrac{1}{\Delta}}\tilde W + P_1-  E}P_1& \ge P_1\pa{\delta\pa{1-\tfrac{1}{\Delta}} +\tfrac{1}{\Delta} }P_1 \notag \\ &  =P_1\pa{\delta\pa{1-\tfrac{1}{\Delta}} +\tfrac{1}{\Delta} }\tilde WP_1,
\end{align}
and
\begin{align}
\bar P_1\pa{\pa{1-\tfrac{1}{\Delta}}\tilde W + P_1-  E}\bar P_1& \ge\bar P_1\pa{\pa{1-\tfrac{1}{\Delta}}\tilde W -  \tfrac 1 2 E\tilde W }\bar P_1 \notag \\
&  \ge\bar P_1\pa{\tfrac  \delta 2 \pa{1-\tfrac{1}{\Delta}}  }\tilde W\bar P_1,
\end{align}
and hence
\begin{align}
 H_N^{(G)}+ P_1-E &\ge \tfrac  \delta 2 \pa{1-\tfrac{1}{\Delta}}  \tilde W\ge \tfrac  \delta 2 \pa{1-\tfrac{1}{\Delta}} >0. \label{barHE1}
\end{align}
As in Lemma~\ref{lemHY}, we conclude that
\begin{align}\label{tWbarHtW2}
\norm{\tilde{W}^{1/2} \pa{H_N^{(G)} + P_1-E}^{-1} \tilde{W}^{1/2}} &\le \tfrac {2} {\delta \pa{1-\tfrac{1}{\Delta}}},
\end{align}
and, for all $\eps \in \R$,
\begin{align}\label{tWbarHtWeps2}
\norm{\tilde{W}^{1/2} \pa{H^{(G)}_N +P_1 -E-i\eps}^{-1}\tilde{W}^{1/2} }\le \tfrac {4} {\delta \pa{1-\tfrac{1}{\Delta}}}.
\end{align}
The estimate \eq{eq:fullCT} follows immediately from  \eq{2a} and \eq{tWbarHtWeps2}.

\end{proof}

\section{Fractional moment estimate on the edge} \label{sec:FMM}

In this section we prove Theorem~\ref{thm:fmlocalization}, using 
 an iterative argument where geometric resolvent identities will be used to decouple random parameters. In principle, this is the familiar strategy in the fractional moments method. However, these resolvent expansions will lead away from the edge into the deep bulk, due to the long range correlations in the $N$-body random potential $V$. We then use (deterministic) exponential Green's function decay in the bulk, i.e.,\ the above Combes-Thomas bounds, to return the path of the iteration to the edge. Somewhat similar arguments have been used before in applications of the fractional moments method to random surface potentials, e.g.\ \cite{BNSS, JL}.

The iteration will naturally lead to graphs with varying geometries, due to the repeated removal of hoppings in the geometric resolvent identities.  A convenient way to handle this is to consider the family of {\it all} Hamiltonians $H_N^{(G)}$ as in   \eq{HNG}, where now for convenience we let  $G=(\mathcal{V}, \mathcal{E},Y)$, where  $(\mathcal{V}, \mathcal{E})$ is a subgraph of $\mathcal X_N$, and $Y$ is a nonnegative bounded potential on $\cV$, and to maximize the relevant quantity over all possible $G$.  Note that given such an operator $H_N^{(G)}$ and $A \subset \cV$, then the decoupled operator 
\beq\label{decop}
H_N^{(G,A)}= \chi_A  H_N^{(G)} \chi_A + \pa{1-\chi_A} H_N^{(G)} \pa{1-\chi_A}
\eeq 
is again of the form given in \eq{HNG}, for an appropriate subgraph  $(\cV, \mathcal{E}_A)$ of $(\mathcal{V}, \mathcal{E})$  and nonnegative bounded potential $Y_A$ on $\cV$.
 All  our results in Section~\ref{sec:prelim}  hold for the operators   $H_N^{(G)}$, with bounds uniform in $G$.

We continue to use $\|\cdot\|$ for $\|\cdot\|_\infty$ and $|\cdot|$ for $\|\cdot\|_1$ on $\mathcal{X}_N$ in the sequel.  We order  $ \mathcal{X}_{N,1}$ by  
\beq\label{eq:conv}
u<v \iff  u_1<v_1.
\eeq
Given $u\in  \mathcal{X}_{N}$ and $k\in\Z$, we use the notation  $u+k=(u_1+k,\ldots,u_N+k)\in \mathcal{X}_{N}$.

Given  $x,y \in \mathcal{X}_{N,1}$, we let (for a fixed $s\in (0,1)$)
\beq 
\tau(y,x) := \sup_G \sup_{E\in I_{1,\delta}}\sup_{ \epsilon>0} \E \left( |\langle \phi_y, (H_N^{(G)}-E-i\epsilon)^{-1} \phi_x \rangle|^s \right).
\eeq
Here the supremum is taken over all  $G=(\mathcal{V}, \mathcal{E},Y)$, where  $(\mathcal{V}, \mathcal{E})$ is a subgraph of $\mathcal X_N$, and $Y$ is a nonnegative bounded potential on $\cV$.  (Or, what is equivalent, the supremum is taken over all  Hamiltonians $H_N^{(G)}$ as in   \eq{HNG}.)
 
 Theorem~\ref{thm:fmlocalization} follows from the following theorem.

\begin{theorem}\label{cor:dec} 
 Fix $s\in (0,1)$.   There exists a constant $C(s,\mu)>0$, depending only on $s$ and $\mu$, independent of $N$,  such that if $\Delta >1$ and $\lambda >0$ satisfy  
 \beq\label{lambdaDeltahyp}
 \lambda \pa{\delta(\Delta -1)}^{{\frac 1 2} } \min \set{1, \pa{\delta(\Delta -1)}^{(\frac 12+\frac 1 s)} }\ge C(s,\mu) ,
 \eeq 
 then for all $N \in \N$  we have 
\beq
\tau(x,y)\le \tfrac { C_\tau} {\lambda^s}\,e^{-\frac m 2 \norm{x-y}} \qtx{for all} x,y \in \mathcal{X}_{N,1},
\eeq
where
\beq\label{Cms}
C_\tau=C_{s,\mu}\pa{1 + \tfrac {8} {\delta(\Delta -1)}}^s\qtx{and}m= s  \log \pa{1 + \tfrac {\delta (\Delta-1)}{8 }},
\eeq 
with $C_{s,\mu}$   a   constant  depending only on $s$ and $\mu$.
\end{theorem}

The following result is the first step in the iterative argument. It establishes a discrete Gronwall-type bound, which will then be used  to prove Theorem~\ref{cor:dec} .

\begin{lemma}\label{lem:k=2}
Let $s\in(0,1)$. There is a constant $C_{s,\mu}$, depending only on  $s$ and the probability distribution $\mu$,   such that, setting $C_\tau$ and $m$ as in \eq{Cms},
 then  for all $N\in \N$ and $x, y \in \mathcal{X}_{N,1}$, \,  if $x<y$ we have
\begin{align}\label{eq:k=2}
&\tau(y,x)\le  \tfrac{C_\tau}{\lambda^s} \pa{\!e^{-m\norm{y-x}}+\!\sum_{w>x} e^{- m\pa{\norm{w-x}}}\tau(y,w)+\e^{-m}\!\!\sum_{w \le x-N}\e^{-m\norm{\pa{x-N}-w}}\tau(y,w)},
\end{align}
and if $y<x$, we have 
\begin{align}\label{eq:k=2'}
&\tau(y,x)\le 
\tfrac{C_\tau}{\lambda^s} \pa{e^{-m\norm{y-x}}+\e^{-m}\sum_{w \ge x+N} e^{- m\norm{w-\pa{x+N}}}\tau(y,w)+\sum_{w<x}\e^{-m\pa{\norm{w-x}}}\tau(y,w) }.
\end{align} 
Finally, there exists a constant $C_\tau^\pr=C^\pr_{s,\mu}$, depending only on  $s$ and the probability distribution $\mu$, such that
\beq\label{eq:k=2=}
\tau(y,x)\le \tfrac{C_\tau^\pr}{\lambda^s} \qtx{for all} x,y\in \mathcal{X}_{N,1}.
\eeq 
\end{lemma}

\begin{proof}
For a given $x\in \mathcal {X}_{N,1}$, let $Q_x=Q_{\set{x_1}}$ denote the indicator of the set 
$S_x=S_{\set{x_1}}$ as in \eq{eq:setSx}, 
and let  $\bar{Q}_x = I - Q_x$. We set  $H^{(G,x)}_N = H^{(G,S_x)}_N$.  Note that $\bar Q_x H^{(G)}_N \bar Q_x$ is a random operator independent of $\omega_{x_1}$.

We will also need a comparison operator $\tilde H^{(G,x)}_N$, which we define, using the terminology introduced at the beginning of Section \ref{sec:prelim}, as
\beq 
\tilde H^{(G,x)}_N= H^{(G,x)}_N+P_1.
\eeq
We write $z=E+i\epsilon$ for $E\in I_1$ and $\epsilon>0$ and note that we have  the estimate \eq{eq:fullCT}, where the  constants $C$ and $m >0$, given in \eq{eq:fullCT2}, are independent of $N$ and $z$.

We will also need   
\begin{lemma} \label{lemma:weak11} Let $s\in (0,1)$.
There is a constant $C_{s,\mu}$, depending only on  $s$ and the probability distribution $\mu$,  such that,
for all $z=E+i\epsilon$, $E\in I_{1,\delta}$, $\epsilon>0$, and $x,w\in \cX_{N,1}$,
\begin{align}\label{eq:weak11}
&\E_{x_1} \left| \scal{\phi_w,\pa{\tilde H^{(G,x)}_N-z}^{-1} \bar Q_x H_N^{(G)} Q_x\pa{ H_N^{(G)}-z}^{-1}\phi_x}\right|^s \notag \\ &  \hskip80pt\le  \frac{C}{\lambda^s}\times \begin{cases}  \e^{-m(\|x-w\|)} &\mbox{if } w>x \\
\e^{-m} \e^{-m\norm{\pa{x-N}-w}} & \mbox{if } x-N \ge w \\
0 & \mbox{if } w\in (x-N,x] \end{cases},
\end{align}
 where $C$ and $m$ are given in \eq{Cms}.
\end{lemma}  

Here we use the conditional expectation $\E_{x_1}(\ldots) = \int_{\R} \ldots \rho(\omega_{x_1}) d\omega_{x_1}$ for averaging over the random variable $\omega_{x_1}$. Below similar notation will also be used for averaging over several of the $\omega_i$.

We postpone the proof of Lemma~\ref{lemma:weak11} until after the completion of the proof of Lemma \ref{lem:k=2}. Consider the case $y>x$. In particular, $y\notin S_x$, so that by the first resolvent identity
\beq
\scal{\phi_y,\pa{H_N^{(G)}-z}^{-1} \phi_x}=-\scal{\phi_y,\pa{H^{(G,x)}_N-z}^{-1} \bar Q_x H_N^{(G)} Q_x \pa{H_N^{(G)}-z}^{-1}\phi_x}.
\eeq
Next, using the first resolvent identity once again, we get
\begin{multline}
\scal{\phi_y,\pa{H_N^{(G)}-z}^{-1} \phi_x}=-\scal{\phi_y,\pa{\tilde H^{(G,x)}_N-z}^{-1} \bar Q_x H_N^{(G)} Q_x \pa{H_N^{(G)}-z}^{-1}\phi_x}\\
-\scal{\phi_y,\pa{H^{(G,x)}_N-z}^{-1} P_1\pa{\tilde H^{(G,x)}_N-z}^{-1} \bar Q_x H_N^{(G)} Q_x \pa{ H_N^{(G)}-z}^{-1}\phi_x}.
\end{multline}
Let us consider both terms on the right hand side. To estimate the first one, we use 
 \eqref{eq:weak11} to get 
\begin{align}
&\E_{x_1}\abs{\scal{\phi_y,\pa{\tilde H^{(G,x)}_N-z}^{-1} \bar Q_x H_N^{(G)} Q_x \pa{H_N^{(G)}-z}^{-1}\phi_x}}^s\le \tfrac{C}{\lambda^s}\, e^{-m\norm{y-x}}.
\end{align}

For the second term, we expand $P_1 = \sum_{w\in \mathcal{X}_{N,1}} |\phi_w\rangle \langle \phi_w|$ and notice that only $w\not\in S_x$ yield non-zero contributions. Using \eqref{eq:weak11} again, we can bound 
\begin{align}\notag
&\E_{x_1}\abs{\scal{\phi_y,\pa{H^{(G,x)}_N-z}^{-1}P_1\pa{\tilde H^{(G,x)}_N-z}^{-1} \bar Q_x H_N^{(G)} Q_x \pa{ H_N^{(G)}-z}^{-1}\phi_x}}^s\\ &  \qquad \le
\frac{C}{\lambda^s}\, \sum_{w>x} e^{- m\norm{x-w}}\abs{\scal{\phi_y,\pa{H^{(G,x)}_N-z}^{-1}\phi_w}}^s \notag \\ &   \hskip50pt  + \frac{C}{\lambda^s}\, \e^{-m}\sum_{w\le x-N} \e^{-m\norm{\pa{x-N}-w}}\abs{\scal{\phi_y,\pa{H^{(G,x)}_N-z}^{-1}\phi_w}}^s.
\end{align} 
This implies that $\tau$ satisfies \eqref{eq:k=2}. The case $y<x$ is dealt with in the same manner, giving  \eqref{eq:k=2'}.

Finally, the bound  \eqref{eq:k=2=} is a consequence of  the weak-$L^1$-bounds in  \cite{AENSS}, which we summarize in Lemma \ref{weak-L1} below.
Given $x,y\in \cX_{N,1}$, let $Q_{xy}$ be the indicator of the set $S_x\cup S_y$ and  $\bar  Q_{xy}=1 -Q_{xy}$. We have
\begin{align} \label{Krein23}
Q_{xy}\pa{ H_N^{(G)}-z}^{-1}\phi_x& =Q_{xy}\pa{ H_N^{(G)}-z}^{-1}Q_{xy} \phi_x \notag \\ 
& =\pa{T_{xy}-\lambda Q_{x}\omega_{x_1}-\lambda Q_{y}\omega_{y_1}}^{-1}Q_x \phi_x,
\end{align}
where the Schur complement $T_{xy}$ is the restriction of 
 \beq 
-\pa{H_N^{(G)}-z-\lambda Q_{x}\omega_{x_1}-\lambda Q_{y}\omega_{y_1}} + H_N^{(G)}\bar Q_{xy}\pa{\bar Q_{xy}\pa{H_N^{(G)}-z} \bar Q_{xy}}^{-1}\!\!\bar Q_{xy}H_N^{(G)}
\eeq
to $\ell_2(S_x\cup S_y)$, and thus an $\omega_{x_1}$- and $\omega_{y_1}$-independent dissipative operator with strictly positive imaginary part (as $\epsilon >0$).
Thus
\begin{align}
\abs{\langle \phi_y, (H_N^{(G)}-z)^{-1} \phi_x \rangle}&=\abs{\langle \phi_y, Q_{y}Q_{xy}(H_N^{(G)}-z)^{-1}Q_{xy} Q_{x}\phi_x \rangle}\notag\\ &= \norm{K_1 Q_y \pa{T_{xy}-\lambda Q_{x}\omega_{x_1}-\lambda Q_{y}\omega_{y_1}}^{-1}Q_{x}K_2}_2 ,
\end{align} 
where $K_1=|\phi_y\rangle\langle\phi_y|$ and $K_2=|\phi_x\rangle\langle\phi_x|$. Thus we can use Lemma~\ref{weak-L1}(ii) below and the standard layer-cake integration argument to conclude the fractional moments bound
\begin{align} \label{eq:layercake}
&\E_{x_1,y_1} \left( |\langle \phi_y, (H_N^{(G)}-E-i\epsilon)^{-1} \phi_x \rangle|^s \right)= \notag \\ &  \ 
\int  \norm{K_1 Q_y \pa{T_{xy}-\lambda Q_{x}\omega_{x_1}-\lambda Q_{y}\omega_{y_1}}^{-1}Q_{x}K_2}_2^s \rho(\omega_{x_1})\rho(\omega_{x_2})\,d \omega_{x_1} d\omega_{x_2} \le \tfrac{C}{\lambda^s},
\end{align} 
with a constant $C$ only depending on $s$ and $\mu$.  This yields \eqref{eq:k=2=}.
\end{proof}

We still have to prove Lemma \ref{lemma:weak11}.

\begin{proof}[Proof of Lemma \ref{lemma:weak11}]
Note that since $x\in S_x$, we have
\beq \label{Krein}
Q_x\pa{ H_N^{(G)}-z}^{-1}\phi_x=Q_x\pa{ H_N^{(G)}-z}^{-1}Q_x\phi_x=-\pa{T_x-\lambda \omega_{x_1}}^{-1}\phi_x,
\eeq 
where the Schur complement $T_x$ is the restriction of 
\beq
-\pa{H_N^{(G)}-z-\lambda \omega_{x_1}} + H_N^{(G)}\bar Q_x\pa{\bar Q_x\pa{H_N^{(G)}-z} \bar Q_x}^{-1}\bar Q_xH_N^{(G)}
\eeq
to $\ell_2(S_x)$, and thus an $\omega_{x_1}$-independent dissipative operator with strictly positive imaginary part.

Using \eqref{Krein}, the left hand side of \eqref{eq:weak11} is equal to $\E_{x_1}\|K_1 (T_x - \lambda\omega_{x_1})^{-1} K_2\|^s_2$ with $K_2=|\phi_x\rangle\langle\phi_x|$ and 
\beq\label{eq:K1b}
K_1= \Big| \phi_w \Big\rangle \Big\langle Q_xH_N^{(G)} \bar Q_x\pa{\tilde H^{(G,x)}_N-\bar{z}}^{-1} \phi_w \Big| .
\eeq

This puts us in a position to use Lemma~\ref{weak-L1}(i) below, combined with the layer-cake integration argument, similar to \eqref{eq:layercake} above, to bound
\beq \label{fmbound}
\E_{x_1}\|K_1 (T_x - \lambda\omega_{x_1})^{-1} K_2\|^s_2 \le C \lambda^{-s} \|K_1\|_2^s \|K_2\|_2^s,
\eeq 
with a constant $C$ only depending on $s$ and $\mu$.

We have $\|K_2\|_2 =1$ and $\|K_1\|_2 = \|Q_xH_N^{(G)} \bar Q_x\pa{\tilde H^{(G,x)}_N-\bar{z}}^{-1} \phi_w\|$. The latter is non-zero only if $w\in S_x^c$, i.e.\ either $w>x$ or $w\le x-N$. 
 
We  estimate $\|K_1\|_2$ by
\begin{align}\label{K12}
\|K_1\|_2& = \|Q_xH_N^{(G)} \bar Q_x\pa{\tilde H^{(G,x)}_N-\bar{z}}^{-1} \phi_w\| \notag \\ 
& \le\norm{Q_xH_N^{(G)} \bar Q_x}\,\norm{\tilde Q_x \pa{\tilde H^{(G,x)}_N-\bar{z}}^{-1} \chi_{\{w\}}},
\end{align} 
where $\tilde Q_x$ is the indicator of the set $\tilde S_x:=\set{y\in \mathcal V:\ \dist_1\pa{y,S_x}\le1}$.  

Since  the number of nearest neighbors of $y\in S_x$ that belong to $S_x^c$ is at most two, we conclude that  $\norm{Q_xH_N^{(G)} \bar Q_x}\le 2\, \tfrac 1 {2\Delta}=\tfrac 1 {\Delta}$. 
The last term in \eq{K12} can be estimated by the Combes-Thomas bound \eq{eq:fullCT} in Corollary \ref{lem1}, which yields
\begin{align}
\norm{\tilde Q_x \pa{\tilde H^{(G,x)}_N-\bar{z}}^{-1} \chi_{\{w\}}}\le  C^\pr e^{-\eta^\pr \dist_1(w,\tilde S_x)},
\end{align}
where $C^\pr$ and $\eta^\pr$ are given in \eq{eq:fullCT2}, so
\beq
\|K_1\|_2\le\ C^{\pr\pr}e^{-\eta^\pr \dist_1(w,\tilde S_x)}, \qtx{where} C^{\pr\pr}=\tfrac {C^\pr}{\Delta}=\tfrac {8} {\delta(\Delta -1)}.
\eeq

Since
\beq
\dist_1\pa{w, \tilde S_x} \ge \begin{cases}
\|(x-w\| -1& \qtx{if} w>x \\
\norm{\pa{x-N}-w} & \qtx{if} x-N\ge w
\end{cases},
\eeq
we obtain
\beq
\|K_1\|_2  \le \begin{cases}
C^{\pr\pr} e^{-\eta^\pr{ \pa{\|x-w\|-1}}} & \text{if} \ w>x \\
C^{\pr\pr} e^{-\eta^\pr \norm{\pa{x-N}-w}}& \text{if} \  x-N\ge w
\end{cases}.
\eeq

The desired bound \eqref{eq:weak11} now follows  from \eqref{fmbound}.  
\end{proof}

We can now prove Theorem~\ref{cor:dec}.

\begin{proof}[Proof of Theorem~\ref{cor:dec}] 
Fix $y\in \mathcal{X}_{N,1}$, let $C_\tau$, $C^\pr_\tau$, and $m$ be as in Lemma~\ref{lem:k=2}.  Let $g(x)=e^{-m\norm{y-x}}$ and  $f(x)=\tau(y,x)$  for  $x\in \mathcal{X}_{N,1}$.  We will denote by $\chi_S$ the indicator function of a set $S$.  It follows from  Lemma~\ref{lem:k=2} that $f(x)\le C^\pr_\tau \lambda^{-s}$ for all  $x\in \mathcal{X}_{N,1}$ and 
$f$ satisfies the integral inequality
\beq\label{eq:inteq}
f(x)\le \tfrac{C_\tau}{\lambda^s} \pa{g+hf}(x) \qtx{for all} x\in \mathcal{X}_{N,1}, x\neq y,
\eeq
where $h$ is the operator on $\ell^2( \mathcal{X}_{N,1})$ whose kernel $h(x,w)$ is given by
\beq\label{eq:h}
h(x,w)=\chi_{(x,\infty)}(w) e^{- m\norm{w-x}}+\chi_{(-\infty,x-N]}(w)\e^{-m}\e^{-m\norm{\pa{x-N}-w}}
\eeq
if $x<y$,   and
\beq\label{eq:h'}
h(x,w)=\e^{-m}\chi_{[x+N,\infty)}(w) e^{- m\norm{w-\pa{x+N}}}+\chi_{(-\infty,x)}(w)e^{- m\norm{w-x}},
\eeq
if $y<x$. Iterating \eqref{eq:inteq} $k$ times, we get
\beq\label{eq:inteq'}
f(x)\le \sum_{j=0}^k \tfrac{C_\tau^{j+1}}{\lambda^{(j+1)s}} \pa{h^jg}(x)+\tfrac{C_\tau^{k+1}}{\lambda^{(k+1)s}} \pa{h^{k+1}f}(x).
\eeq
To bound the first term, let $\hat g(x)=e^{-
\frac m 2\norm{x-y}}$ and consider the case $x<y$. Then
\begin{align}
&\pa{h\hat g}(x) =\sum_{w>x} e^{- m\norm{w-x}}e^{-\frac m 2\norm{w-y}}+\e^{-m}\sum_{w\le x-N}\e^{-m\norm{\pa{x-N}-w}}e^{-\frac m 2\norm{w-y}} \notag \\
&\quad  \le \sum_{w>x} e^{- \frac m 2\norm{w-x}}e^{-\frac m 2\norm{x-y}}+\e^{-m}\sum_{w\le x-N}\e^{-m\norm{\pa{x-N}-w}}e^{-\frac m 2\norm{x-y}}\le  \hat C \, \hat{g}(x),
\end{align}
where
\begin{align}
\hat C& =\sum_{w>x} e^{- \frac m 2\norm{w-x}} +\e^{-m} \sum_{w\le x-N}\e^{-m\norm{\pa{x-N}-w}} \notag \\ &  \le \sum_{r=1}^\infty e^{- \frac m 2 r} + \sum_{r=1}^\infty \e^{-mr} \le 2\sum_{r=1}^\infty e^{- \frac m 2 r} = 2 e^{- \frac m 2}\pa{1-   e^{- \frac m 2} }^{-1}.
\end{align}

The same can be verified for $x>y$. By induction we find that $\pa{h^jg}(x)\le C^j\hat g(x)$. 

On the other hand, by the a-priory bound \eqref{eq:k=2=} we have $\|f\|_\infty\le C^\pr_\tau \lambda^{-s}$,
and 
\begin{align}
 e^{- m\norm{w-x}}+\e^{-m}\sum_{w\le x-N}\e^{-m\norm{\pa{x-N}-w}}\le 2\ \sum_{r=1}^\infty \e^{-mr}  =  e^{-  m }\pa{1-   e^{-  m }  }^{-1},
 \end{align}
with a similar bound in the case given by \eq{eq:h'}. Thus , setting $\tilde C=  \pa{e^{ \frac m 2}-   e^{- \frac m 2} }^{-1}$,
 the last term in \eqref{eq:inteq'} can be estimated by
\beq
\pa{h^{k+1}f}(x)\le \tfrac{C^\pr_\tau}{\lambda^s}\tilde C^{k+1}\e^{-\frac m 2 (k+1)},
\eeq
again proven inductively.
Putting these two bounds together, we get 
\beq\label{eq:inteqfin}
f(x)\le \tfrac{C_\tau}{\lambda^s}\sum_{j=0}^k\pa{\tfrac{C_\tau\hat C }{\lambda^s}}^{j} \hat g(x)+ \tfrac{C^\pr_\tau}{\lambda^s} \tfrac{C_\tau^{k+1}}{\lambda^{(k+1)s}} \tilde C^{k+1}e^{-\frac m 2 (k+1)}.
\eeq

We have
\beq
{C_\tau\hat C } =\tfrac {2 C_{s,\mu}\pa{1 + \tfrac {8} {\delta(\Delta -1)}}^s} {\pa{\pa{1 + \tfrac  {\delta(\Delta -1)}{8}}^{\frac s 2}-1}} \qtx{and} {C_\tau\tilde C } =\tfrac {C_{s,\mu}\pa{1 + \tfrac {8} {\delta(\Delta -1)}}^s}{\pa{\pa{1 + \tfrac  {\delta(\Delta -1)}{8}}^{\frac s 2}-{\pa{1 + \tfrac  {\delta(\Delta -1)}{8}}^{-\frac s 2}}}},
\eeq
so
\begin{align}\label{boundDL}
\max\set{{C_\tau\hat C } ,{C_\tau\tilde C } }\le {C^\pr_{s,\mu}}\max \set{\pa{\delta(\Delta -1)}^{-(1+s)} ,\pa{\delta(\Delta -1)}^{-{\frac s 2} }},
\end{align}
where the constant $C^\pr_{s,\mu}$ depends only on $s$ and $\mu$.

Taking $(\Delta,\lambda)$ so the right hand side of \eq{boundDL} is $\le \frac {\lambda{^s}}2$, we get
\beq
f(x)\le \tfrac{2 C_\tau}{\lambda^s} \hat g(x) +   \tfrac{C^\pr_\tau}{\lambda^s} 
\e^{-\frac m2 (k+1)}
\eeq
The choice $k+1=\|x-y\|$ yields
\beq\label{eq:inteqfin'}
\tau(y,x)=f(x)\le  \tfrac{C^{\pr\pr}_{s,\mu} C_\tau}{\lambda^s}  e^{-\frac m 2 \norm{x-y}},
\eeq
where the constant $C^{\pr\pr}_{s,\mu}$ depends only on $s$ and $\mu$.
\end{proof}

Finally, for completeness, we state the weak-$L^1$-bounds which we have used above, see Lemma~3.1 and Proposition~3.2 of \cite{AENSS}. Here $\|\cdot\|_2$ denotes Hilbert-Schmidt norm and $|\cdot|$ Lebesgue measure in dimension one and two, respectively.
 
\begin{lemma}[\cite{AENSS}] \label{weak-L1}
There exists a universal constant $C <\infty$ such that the following holds:

(i) For arbitrary separable Hilbert spaces $\mathcal{H}$ and $
\mathcal{H}_1$, arbitrary maximally dissipative operators $A$ with strictly positive imaginary part, and arbitrary Hilbert-Schmidt operators $M_1:\mathcal{H} \to \mathcal{H}_1$ and $M_2: \mathcal{H}_1 \to \mathcal{H}$,
\beq
 \left| \{v\in \R: \|M_1(A-v)^{-1}M_2\|_2 >t\} \right| \le C \|M_1\|_2 \|M_2\|_2 \frac{1}{t}.
 \eeq
 
 (ii) Moreover, for $A$, $M_1$ and $M_2$ as above and arbitrary non-negative operators $U_1$ and $U_2$,
 \beq
 \left| \left\{ (u,v) \in [0,1]^2: \|M_1 U_1^{1/2} (A-v_1U_1-v_2U_2)^{-1} U_2^{1/2} M_2\|_2 > t \right\} \right| \le 2 C \|M_1\|_2 \|M_2\|_2 \frac{1}{t}.
 \eeq
 \end{lemma}

\section{Wegner estimate for the droplet spectrum}\label{sec:Wegner}

\subsection{``Boxes'' in $\mathcal X^N$}
A useful concept in the theory of Schr\"odinger operators on $\ell^2\pa{\Z^d}$ are their restrictions to  finite boxes $[-M,M]^d\cap\Z^d$, $M\in \N$, or translates of this set. Information concerning the number of eigenvalues on a small interval for such restrictions as well as their independence play an important role in the analysis of the Anderson model on $\Z^d$. Let us now introduce the counterpart of this concept for the $N$-body operators $H_N$ analyzed here. In fact, much of the following considerations depend on an interplay between two types of ``boxes'', one of them, denoted by $\Lambda_M$, being boxes along the edge $\mathcal{X}_{N,1}$, and the other, denoted by $S_M$, being boxes (or rather unions of ``strips'') in the full lattice $\mathcal{X}_N$, which will we used to create required independence properties.

More precisely, throughout the remainder of this work we will consider finite volume operators $H_N^{(L)}$ defined in (\ref{eq:H_N}). Thus all the point sets to be introduced will generally be considered as subsets of $\mathcal{X}_N^{(L)}$, without always specifying this in our notation.

Finite boxes along the edge are given by the subsets 
\beq \label{eq:LambdaM}
\La_M(x)=[x-M,x+M] = \{(y_1,y_1+1,\ldots,y_1+N-1): x_1-M \le y_1 \le x_1+M\}
\eeq 
of $\mathcal{X}_{N,1}$, for some $x\in\mathcal X_{N,1}$, where we use the convention introduced before Lemma \ref{lem:k=2}.  For such $\La_M(x)$, we define a subset of $\Z$ as 
\beq\label{eq:PsiLa}
\Psi\pa{\La_M(x)}=\set{x_1-M,\ldots,x_1+M},
\eeq
i.e.\ the first components of sites in $\La_M(x)$, thus introducing a one-to-one correspondence between intervals in $\Z$ and ``intervals'' in $\mathcal{X}_{N,1}$. For a given $\Psi\subset \Z$, we define a subset $S_\Psi$ of $\mathcal  X_N$ as 
\beq\label{eq:setLa'}
S_\Psi:=\set{u\in \mathcal  X_N:\ u_j\in \Psi\mbox{ for some } j\in\{1,\ldots,N\}},
\eeq
meaning that $S_\Psi$ contains all the sites $u$ on which the random potential $V_{\omega}(u)$ depends on one of the random parameters $\omega_j$, $j\in \Psi$.
Note that $S_x=S_{\set{x_1}}$ for the notation used in Section~\ref{sec:FMM}.

 In particular, 
 \beq\label{eq:setSM}
 S_M(x):= S_{\Psi\pa{\La_M(x)}} = S_{\{x_1-M,\ldots,x_1+M\}}
 \eeq
  are the lattice sites at which the potential depends on the random parameters $\omega_j$ with $x_1-M\le j \le x_1+M$. Note that 
\beq\label{eq:LaPsi}
\pa{S_M(x)}_1:= S_M(x) \cap \mathcal X_{N,1}=   [x-(M+N-1),x +M] 
\eeq

\subsection{Wegner estimate}

We are now ready to formulate and prove a Wegner estimate in the droplet spectrum for the restrictions of $H_N^{(L)}$ to the sets $S_M(x)$. 

For some $x\in \mathcal{X}_{N,1}$ let $H_{S_M(x)}^{(L)}$ be the restriction of $H_N^{(L)}$ to $S_M(x)$ (assuming that $L$ is sufficiently large so that $S_M(x) \cap \mathcal{X}_N^{(L)}$ in non-empty). For a subset $\Psi \subset \Z$ we will denote by $\P_{\Psi}$ the conditional probability with respect to the random variables $\omega_i$, $i\in \Psi$, with all other $\omega_i$ fixed.

Our main result of this section is the following Wegner estimate.

\begin{lemma} \label{lem:Weg}
Let $C_W=C_W(\delta,\Delta)= \pa{1 +\tfrac{\sqrt{2}}{\delta \pa{\Delta -1 }} }$. Then 
\beq \label{eq:Weg}
\P_{\{x_1-M,\ldots,x_1+M\}} \left\{ \sigma(H_{S_M(x)}^{(L)}) \cap I \neq \emptyset \right\} \le C_W  \lambda^{-1} \|\rho\|_{\infty}(2M+1) (2M+N)|I|  
\eeq 
for all $x\in \mathcal{X}_{N,1}$, $N, M \in \N$ and subintervals $I \subset I_{1,\delta}$ with $|I| \le 2\delta \pa{1-\tfrac{1}{\Delta}}/C_W(\delta,\Delta)$.
\end{lemma}

    The main idea behind our proof of Lemma~\ref{lem:Weg} is to reduce it to a Wegner estimate for the Schur complement of $H_{S_M(x)}^{(L)}$ with respect to the decomposition of the configuration space $S_{\Psi(\Lambda_M(x))}$ into edge and bulk. To this end, let $Q= Q_{\pa{S_M(x)}_1 }$ be the indicator function of $\pa{S_M(x)}_1$ and $\bar{Q} = I-Q$ in $\ell^2(S_M(x))$. 

Let $E$ be the center of a subinterval $I \subset I_{1,\delta}$ and set $D=H_{S_M(x)}^{(L)}-E$, $A=QDQ$, $B=\bar{Q}D\bar{Q}$ and $V=QD\bar{Q} = -Q\mathcal{L}\bar{Q}/2\Delta$, where $\mathcal{L}$ is the Laplacian defined in \eqref{eq:lapl}. Thus
\beq
H_{S_M(x)}^{(L)}-E = \begin{pmatrix} A & V \\ V^* & B \end{pmatrix}.
\eeq
$B$ is invertible with $\|B^{-1}\| \le (\delta(1-1/\Delta))^{-1}$ by \eqref{eq:klow} and  $\|V\| \le \sqrt{2}/\Delta$, which uses the fact that each $u\in \mathcal{X}_{N,1}$ has at most two nearest neighbors. 

Thus $H_{S_M(x)}^{(L)}-E$ is invertible if and only if the Schur complement
\beq \label{eq:Schur2}
K_E = D/B = Q(H_{S_M(x)}^{(L)}-E)Q - VB^{-1}V^*
\eeq
is invertible and $Q(H_{S_M(x)}^{(L)}-E)^{-1} Q = K_E^{-1}$. With this we can reduce Lemma~\ref{lem:Weg} to the following Wegner estimate for $K_E$:

\begin{lemma} \label{lem:WegK}
For all $E\in I_{1,\delta}$, $M,N \in \N$, and $\epsilon>0$ we have 
\beq \label{eq:WegK}
\P_{\{x_1-M\ldots,x_1+M\}} \left\{ \dist(\{0\}, \sigma(K_E)) < \epsilon \right\} \le 2\eps  \lambda^{-1}  \|\rho\|_{\infty} (2M+1)(2M+N) .
\eeq 
\end{lemma}

We will first use this to complete the 

\begin{proof}[Proof of Lemma~\ref{lem:Weg}]
It suffices to show the claim for $|I|$ sufficiently small. If $\dist (\{0\}, \sigma(K_E)) \ge \epsilon$, then $\|K^{-1}\| \le 1/\epsilon$. That $K$ is invertible implies that $D = H_{S_M(x)}^{(L)}-E$ is invertible and, by a standard fact, that 
\beq
D^{-1} = \begin{pmatrix} K_E^{-1} & -K_E^{-1} V B^{-1} \\ - B^{-1} V^* K_E^{-1} & B^{-1} + B^{-1} V^* K_E^{-1} V B^{-1} \end{pmatrix}.
\eeq
Thus $\|D^{-1}\|$ can be bounded in terms of the norms of the four matrix elements and we can conclude form the bounds provided above that  that for $\eps \le  \delta \pa{1-\tfrac{1}{\Delta}}$ we have $\|D^{-1} \| \le \frac{C(\delta,\Delta)}{\epsilon}$, where $C(\delta,\Delta)= \pa{1 +\tfrac{\sqrt{2}}{\delta \pa{\Delta -1 }} }$.
Therefore $\sigma(H_{S_M(x)}^{(L)}-E) \cap (-\epsilon/C(\delta,\Delta), \epsilon/C(\delta,\Delta)) = \emptyset$.  Using  Lemma~\ref{lem:WegK}, we conclude  that, if $|I| \le 2\delta \pa{1-\tfrac{1}{\Delta}}/C(\delta,\Delta)$, then  $\sigma(H_{S_M(x)}) \cap I = \emptyset$ holds with conditional probability at least $1 - C(\delta,\Delta) \lambda^{-1} \|\rho\|_{\infty} (2M+1)(2M+N)|I|/2$, yielding Lemma~\ref{lem:Weg}.
\end{proof}

To prove Lemma~\ref{lem:WegK} we use the following lemma from \cite{Stollmann}. To state it, let $\mu$ be a probability measure on $\R$ and, for any $\eps>0$, $s(\mu,\eps) := \sup\{\mu([\alpha,\beta]): \beta-\alpha \le \eps\}$. Note that in the case that $\mu$ has a bounded density $\rho$, as in our application, we have $s(\mu,\eps) \le \|\rho\|_{\infty} \eps$. 

Also let $J$ be a finite index set and denote by $\mu^J$ the $J$-fold product measure of $\mu$ on $\R^J$.

\begin{lemma}[\cite{Stollmann}] \label{Stollmann} Consider a monotone function $\Phi$ on $\R^J$ which satisfies $\Phi(q+te) - \Phi(q) \ge t$ for $e=(1,1,\ldots,1) \in \R^J$ and all $t>0$. Then, for any open interval $I\subset \R$, we have
\beq \mu^J\{q: \Phi(q) \in I\} \le |J| s(\mu,|I|).
\eeq
\end{lemma}

\begin{proof}[Proof of Lemma~\ref{lem:WegK}]
Let $E_1 \le E_2 \le \ldots $ be the at most $| \pa{S_M(x)}_1 | = 2M+N$ eigenvalues of $K_E$ (fewer if $[x-(M+N-1),x +M]$ is not fully contained in $\mathcal{X}_N^{(L)}$), counted with multiplicity, which we consider as functions of  $\tilde{\omega} = (\omega_{x_1-M}, \ldots, \omega_{x_1+M})$,  
with all other $\omega_i$ fixed. Let $e=(1,1,\ldots,1)$ as in Lemma~\ref{Stollmann} with $J= 2M+1$. Note that $-VB^{-1}V^*$ in \eqref{eq:Schur2} is monotone increasing in all components of $\tilde{\omega}$ in quadratic form sense. Thus, for all $t>0$,
\begin{eqnarray}
K_E(\tilde{\omega} + te) - K_E(\tilde{\omega}) & \ge & Q (H_{S_M(x)}(\tilde{\omega}+te) - H_{S_M(x)}(\tilde{\omega})) Q \nonumber \\
& = & \lambda Q (V_{\tilde{\omega}+te} - V_{\tilde{\omega}}) Q \;\ge\; \lambda tI,
\end{eqnarray}
where we have used in the last step that each site $u$ in $\pa{S_M(x)}_1 $ has at least one component in  $\{x_1-M,\ldots, x_1+M\}$, so that $V_{\tilde{\omega}+te}(u)-V_{\tilde{\omega}}(u) \ge 1$. Thus we have $E_n(\tilde{\omega}+te) - E_n(\tilde{\omega}) \ge \lambda t$ for all $t>0$ and all $n$ by the min-max principle and can apply Lemma~\ref{Stollmann}  with $\Phi = \lambda^{-1} \ E_n$ to get
\beq
\P_{ \{x_1-M,\ldots, x_1+M\}} \{ \tilde{\omega}: E_n(\tilde{\omega}) \in (-\epsilon, \epsilon) \} \le\ 2\epsilon   \lambda^{-1} \|\rho\|_{\infty} (2M +1).
\eeq 
Using this for each one of the eigenvalues $E_n$ yields \eqref{eq:WegK}.
\end{proof}

\section{Localization on a pair of ``boxes"}\label{sec:pairofboxes}

Consider $i, j \in \Z$, let  $x$ and $y$ be the points in $\mathcal{X}_{N,1}$ with first components $x_1=i$ and $y_1=j$, i.e., 
$x=\pa{i,\ldots i+N-1}$ and  $y=\pa{j,\ldots j+N-1}$, and let   $S_M(x)$ and $S_M(y)$ be as in \eq{eq:setSM}. Note that  (recall \eq{eq:LaPsi})
\begin{align} \label{distxy2}
&\dist_{\infty}(\pa{S_M(x)}_1,\pa{S_M(y)}_1)= \notag \\
 & \quad \dist_{\infty}([x-(M+N-1),x +M] , [y-(M+N-1),y +M]   )= \abs{i-j} - 2M -N +1.
\end{align} 
We assume
\beq\label{Mijb23}
M=M(i,j) : = \lfloor{\tfrac{\abs{i-j}}4}\rfloor= \lfloor{\tfrac{\norm{x-y}}4}\rfloor >  N,
\eeq 
so, in particular,
\begin{align} \label{distxy}
&\dist_{\infty}(\pa{S_M(x)}_1,\pa{S_M(y)}_1)> M+1.
\end{align}

 The operators $H_{S_M(x)}^{(L)}$ and $H_{S_M(y)}^{(L)}$ are {\it not} statistically independent, due to the long range correlations in the $N$-body random potential $V_{\omega}$. But we can achieve a sufficient amount of conditional independence for our purposes by deleting the random potential from $S_M(x) \cap S_M(y)$. Thus, denoting the indicator function of the latter set by $Q_{x,y}$, we define
\beq\label{deftildeH}
\tilde H_{S_M(u)}^{(L)} = H_{S_M(u)}^{(L)} - \lambda Q_{x,y}V_{\omega}
\eeq
for $u \in \{x,y\}$.     If $w\in S_M(y)$ and $z\in \pa{S_M(x)}_1$ (and similarly for $w\in S_M(x)$ and $z\in \pa{S_M(y)}_1$) it follows from \eqref{distxy} that $\|w-z\|_{\infty} >M+1$, so\beq \label{distSLa}
\dist_{\infty}(S_M(x) \cap S_M(y)), \pa{S_M(x)}_1\cup \pa{S_M(y)}_1) > M+1.
\eeq 
Thus the operators $\tilde H_{S_M(u)}^{(L)}$ and $H_{S_M(u)}^{(L)}$ differ only at sites at distance $> M+1 $ from the edge $\pa{S_M(u)}_1$, with  $u=x,y$. In the second part of the following lemma we show that this implies that within the droplet spectrum their eigenvalues are exponentially close in $M$. 

If $A$ is a self-adjoint operator and $I\subset \R$, we set $\sigma_I(A)=\sigma(A)\cap I$.

\begin{lemma} \label{lemsepboxes}   Assume 
\beq\label{Mijb9}  
M=M(i,j) : = \lfloor{\tfrac{\abs{i-j}}4}\rfloor= \lfloor{\tfrac{\norm{x-y}}4}\rfloor \ge   2 N+ 2 \beta, \mqtx{where}  \beta \;\;\text{is given in  \eq{finiteXXZ}}.
\eeq 
  Then
\begin{enumerate}
\item $\tilde H_{S_M(x)}^{(L)}$ does not depend on the random variables $\omega_\ell$, $\ell=j-M,\ldots,j+M$, and $\tilde H_{S_M(y)}^{(L)}$ does not depend on $\omega_{\ell}$, $\ell=i-M,\ldots,i+M$.

\item We have 
\beq\label{eq:spcomp}
\dist(E,\sigma(\tilde{H}_{S_M(x)}^{(L)}))\le C_1  e^{-\frac 3 2\eta_1 M}
\eeq
for all $E\in \sigma_{ I_{1,\frac \delta 2}}(H_{S_M(x)}^{(L)})$, where
\beq
\eta_1=  \log \pa{1 + \tfrac {\delta (\Delta-1)}{16}}\qtx{and} C_1= \frac {128}{\eta_1 \delta^2(\Delta-1)^2}.
\eeq
\item Let $0<\eps\le \eps_1=    \delta \pa{1-\tfrac{1}{\Delta}}/(4 C_W(\delta/4,\Delta)) $.  We have
\begin{align}\label{tildeweg}
\P\set{\dist \pa{\sigma_{ I_{1,\frac \delta 4}}(\tilde{H}_{S_M(x)}^{(L)}) ,\sigma_{ I_{1,\frac \delta 4}}(\tilde{H}_{S_M(y)}^{(L)}) }\le \eps}\le  C C_W   \lambda^{-1} \|\rho\|_{\infty} M^3 \eps,
\end{align}
where $C$ is an independent constant and $C_W=C_W(\frac \delta 4,\Delta)$ is as in Lemma~\ref{lem:Weg}.

\item Suppose $ {C_1}  e^{-\frac 3 2\eta_1 M} < \frac \delta 4 \pa{1 - \frac 1 \Delta}$ and 
 $0<\eps\le  \frac { \delta \pa{1-1/\Delta}}{4 C_W(\delta/4,\Delta)}- 2 {C_1}  e^{-\frac 3 2\eta_1 M} $.
Then
\begin{align}\label{tildeweg2}
&\P\set{\dist \pa{\sigma_ {I_{1,\frac \delta 2}}({H}_{S_M(x)}^{(L)}) ,\sigma_ {I_{1,\frac \delta 2}}({H}_{S_M(y)}^{(L)})  } \le \eps} \notag \\ &   \hskip100pt \le  C C_W \lambda^{-1} \|\rho\|_{\infty} M^3 \pa{\eps+  2 {C_1}  e^{-\frac 3 2\eta_1 M}}.
\end{align}

\end{enumerate}
\end{lemma}

\begin{proof} 
Part (i) follows immediately from \eq{deftildeH} and \eqref{eq:setLa'}. We will abbreviate $S_x := S_M(x)$, $H_{S_x} := H_{S_x}^{(L)}$ and $\tilde{H}_{S_x} := \tilde{H}_{S_x}^{(L)}$. To prove part  (ii) we proceed by Schur complementation with respect to the decomposition of $S_x$ into  $\pa{S_x}_1 :=S_x \cap \mathcal{X}_{N,1} $  and $S_x \cap \bar{\mathcal{X}}_{N,1}$. Let $\pa{\cdot}_1$ and $\overline{\pa{\cdot}}_1$ denote the corresponding restrictions, define
\beq
K_E=\pa{H_{S_x}-E-H_{S_x}\overline{\pa{H_{S_x}-E}}_1^{-1} H_{S_x}}_1,
\eeq
and let $\tilde K_{E}$ be defined analogously, with $H_{S_x}$ replaced by $\tilde H_{S_x}$. 
We first observe that it follows from \eq{eq:klow} that  for any $E\in I_{1,\frac \delta 2}$ the operators $K_{E}$ and $\tilde{K}_E$ are well defined and bounded.

Using the resolvent identity  we get 
\beq
 K_E-\tilde{K}_E=P_{1,x}\tilde{H}_{S_x}\bar P_{1,x} \overline{\pa{\tilde{H}_{S_x}-E}}_1^{-1} Q_{x,y} \lambda V \overline{\pa{H_{S_x}-E}}_1^{-1} \bar P_{1,x}H_{S_x} P_{1,x}, 
 \eeq
where $P_{1,x}$ is the indicator function of $\pa{S_x}_1$, and $\bar P_{1,x} =\chi_{S_x} - P_{1,x}$. Letting $H_0=H-\lambda V$, we have
\beq
\lambda V \overline{\pa{H_{S_x}-E}}_1^{-1}= \bar P_{1,x} -  \overline{\pa{\pa{H_{0}}_{S_x}-E}}_1  \overline{\pa{H_{S_x}-E}}_1^{-1},
\eeq
and hence
\beq \label{KK1}
K_E-\tilde{K}_E = - P_{1,x}\tilde{H}_{S_x}\bar P_{1,x} \overline{\pa{\tilde{H}_{S_x}-E}}_1^{-1} Q_{x,y}\overline{\pa{\pa{H_{0}}_{S_x}-E}}_1  \overline{\pa{H_{S_x}-E}}_1^{-1} \bar P_{1,x}H_{S_x} P_{1,x},
\eeq
where we used that  $Q_{x,y}\chi_{\pa{S_x}_{1,2}} =0$ for $\pa{S_x}_{1,2}=\set{u\in S_x \cap \mathcal{X}_{N,2}; \ \dist_1 (u,\pa{S_x}_1)=1 }$, and
\beq 
Q_{x,y} \bar P_{1,x}H_{S_x} P_{1,x}=Q_{x,y}\chi_{\pa{S_x}_{1,2}} \bar P_{1,x}H_{S_x} P_{1,x}.
\eeq

   Clearly,
\beq\label{KK12}
Q_{x,y}  \overline{\pa{\pa{H_{0}}_{S_x}-E}}_1=Q_{x,y} \overline{\pa{\pa{H_{0}}_{S_x}-E}}_1\tilde Q_{x,y},
\eeq
where $\tilde Q_{x,y}$ is the indicator of the set $\set{u\in \mathcal X_N:\ \dist_1\pa{u,\supp Q_{x,y}}\le1}$.
 Next, we observe that it follows from \eq{eq:H_N},  \eq{dropspec}, and \eq{Mijb9} that  
\beq\label{KK123}
 \norm{\pa{H_0}_{S_x}-E}\le  2N+2 \beta \le M,
 \eeq 
 and, since vertices in $\mathcal{X}_{N,1}$ have at most  two next neighbors, we have
 \beq\label{KK124}
\norm{\bar P_{1,x}H_{S_x} P_{1,x}}=\norm{\bar P_{1,x}\tilde H_{S_x} P_{1,x}}\le \tfrac{\sqrt{ 2}}\Delta. 
\eeq  

Combining \eq{KK1}, \eq{KK12}, \eq{KK123}, and \eq{KK124} we obtain the bound 
\beq
\norm{K_E-\tilde{K}_E}  \le \tfrac{2M}{\Delta^2}  \norm{\chi_{\pa{S_x}_{1,2}}\overline{\pa{\tilde{H}_{S_x}-E}}_1^{-1} Q_{x,y}}\norm{\tilde Q_{x,y}\overline{\pa{H_{S_x}-E}}_1^{-1}\chi_{\pa{S_x}_{1,2}}}. 
\eeq

Using \eq{2} in Corollary~\ref{lem1}, we have
\begin{align}
 &\norm{\chi_{\pa{S_x}_{1,2}}\overline{\pa{\tilde{H}_{S_x}-E}}_1^{-1} Q_{x,y}}\le  C_1 e^{-\eta_1 \dist_1\pa{\pa{S_x}_{1,2}, S_M(x) \cap S_M(y)}} \notag \\  & \qquad \qquad\le
C_1 e^{-\eta_1\pa{ \dist_1\pa{\pa{S_x}_{1}, S_M(x) \cap S_M(y)}-1}}\le C_1 e^{-\eta_1 (M+1)},
\end{align}
where $C_1 =C(\Delta, \frac \delta 2)$ and $\eta_1=\eta(\Delta,1, \frac \delta 2)$ are as in \eq{Ceta}, and we used \eq{distSLa}.
Moreover,\\ $ e^{\eta_1 }\norm{\tilde Q_{x,y}\overline{\pa{H_{S_x}-E}}_1^{-1}\chi_{\pa{S_x}_{1,2}}}$ satisfies the same bound.  Thus 
\begin{align}\label{KEKE}
\norm{K_E-\tilde{K}_E}& \le  \frac{2M}{\Delta^2}     C_1^2 e^{-2\eta_1 M}\ \le \frac{4}{\e \eta_1\Delta^2}     C_1^2 e^{-\frac 3 2\eta_1 M}\le \frac{2}{\eta_1 \Delta^2}   C_1^2 e^{-\frac 3 2\eta_1 M}.
\end{align}

Now, suppose $E\in \sigma_{I_{1,\frac \delta 2}}(H_{S_x}) $, so  $0\in\sigma\pa{K_{E}}$. Letting $C^\pr= \frac{2}{\eta_1 \Delta^2} C_1^2 e^{-\frac 3 2\eta_1 M}$, we deduce from \eq{KEKE} that $\left[-C^\pr,C^\pr \right]\cap \sigma\pa{\tilde K_{E}}\neq \emptyset$,
and hence $\tfrac 1 {C^\pr}\le \|\tilde K_{E}^{-1}\|\le\|(\tilde H_{S_x}-E)^{-1}\|$,
which implies \eqref{eq:spcomp}.

To prove (iii), note that  $\tilde{H}_{S_M(x)}^{(L)}= \pa{ \tilde{H}_{S_M(x)}^{(L)}{+}P_{1,x}} {-}P_{1,x}$ and $\sigma_ {I_{1,\frac \delta 4}}\pa{  \tilde{H}_{S_M(x)}^{(L)}{+}P_{1,x}}=\emptyset$.
Since $\tr P_{1,x}{=} \abs{\pa{S_M(x)}_1}{=} 2M +N $, we conclude that $\tilde{H}_{S_M(x)}^{(L)}$ has at most $2M +N$ eigenvalues in  $I_{1,\frac \delta 4}$. Thus it follows from 
 Lemma~\ref{lem:Weg} and part (i) that
 \begin{align}
& \P\set{\dist \pa{\sigma_ {I_{1,\frac \delta 4}}(\tilde{H}_{S_M(x)}^{(L)}) ,\sigma_ {I_{1,\frac \delta 4}}(\tilde{H}_{S_M(y)}^{(L)}) }\le \eps} \notag \\ &  \qquad \le 
2C_W  \lambda^{-1} \|\rho\|_{\infty}(2M+1) (2M+N)^2  \eps \le C C_W  \lambda^{-1} \|\rho\|_{\infty} M^3 \eps,
 \end{align}
where  $C_W=C_W(\frac \delta 2,\Delta)$ is as in Lemma~\ref{lem:Weg}.

Finally, we prove part (iv).  Under the hypotheses, for $E\in \sigma_{ I_{1,\frac \delta 2}}(H_{S_M(x)}^{(L)})$ we have, for  $u=x,y$,
\begin{align}
\dist(E,\sigma(\tilde{H}_{S_M(u)}^{(L)}))=\dist(E,\sigma_{ I_{1,\frac \delta 4}}(\tilde{H}_{S_M(u)}^{(L)})),
\end{align}
so, using \eq{eq:spcomp}, we get 
\begin{align}\
&\dist \pa{\sigma_{ I_{1,\frac \delta 4}}(\tilde{H}_{S_M(x)}^{(L)}) ,\sigma_{ I_{1,\frac \delta 4}}(\tilde{H}_{S_M(y)}^{(L)}) } \ \notag \\ & \notag 
\ \le \dist \pa{\sigma_{ I_{1,\frac \delta 4}}(\tilde{H}_{S_M(x)}^{(L)}) ,\sigma_{ I_{1,\frac \delta 2}}({H}_{S_M(x)}^{(L)}) }
 + \dist \pa{\sigma_{ I_{1,\frac \delta 2}}({H}_{S_M(x)}^{(L)}) ,\sigma_{ I_{1,\frac \delta 2}}({H}_{S_M(y)}^{(L)}) } \\  \notag &
 \hskip100pt+\dist \pa{\sigma_{ I_{1,\frac \delta 2}}({H}_{S_M(y)}^{(L)}) ,\sigma_{ I_{1,\frac \delta 4}}(\tilde{H}_{S_M(y)}^{(L)}) }\\ &  \ 
\le\dist \pa{\sigma_{ I_{1,\frac \delta 2}}({H}_{S_M(x)}^{(L)}) ,\sigma_{ I_{1,\frac \delta 2}}({H}_{S_M(y)}^{(L)}) } + 2C_1  e^{-\frac 3 2\eta_1 M} .
\end{align}
 Thus  \eq{tildeweg2}  follows from  \eq{tildeweg}.
\end{proof}

From now on  we assume \eq{lambdaDeltahyp} is satisfied, so we have the conclusions of Theorem~\ref{cor:dec}.

For the pairs $i, j \in \Z$ and corresponding $x,y \in \mathcal{X}_{N,1}$ as defined at the beginning of this section, let $Q_{i}$ and $Q_{j}$  be the indicator functions of the subsets $S_{\set{i}}\cap \mathcal{X}_N^{(L)}$ and $S_{\set{j}}\cap \mathcal{X}_N^{(L)}$, respectively, defined in \eqref{eq:setLa'}. 
In what follows, we will use the shorthand notations $\La_x$ and $S_x$ for $\La_M(x)$ and $S_M(x)$, defined in \eqref{eq:LambdaM} and \eq{eq:setSM}, and similarly for $y$. 
 For  $S=S_x$ or $S_y$, we set (recall \eq{decop})
\beq\label{eq:decH}
 H^{(L)}_{S}=Q_{S}H_N^{(L)} Q_{S}+\bar Q_{S}H_N^{(L)}\bar  Q_{S} \qtx{and}G^{S}_E=\pa{ H^{(L)}_{S}-E}^{-1}.
\eeq  
 We also let
 \begin{align} \label{GammaS}
\Gamma_{S}& = H^{\pa{L}}_N-H^{\pa{L}}_{S}= \chi_{\partial_- S}\Gamma_{S} \chi_{\partial_+S}, \qtx{where} \notag \\ \notag
\partial_- S &= \set{u \in S: \abs{u-v}_1=1 \sqtx{for some} v \in \mathcal{X}_N^{(L)}\setminus S},\\ 
\partial_+S &= \set{v \in \mathcal{X}_N^{(L)}\setminus S: \abs{u-v}_1=1 \sqtx{for some} u \in  S },
 \end{align}
 and note that $\norm{\Gamma_{S} }\le \tfrac N \Delta$.

In this section we let  $H_{S}=H^{(L)}_{S}$.

\begin{lemma}\label{lem:pairb} Let  $i,j\in \Z$ ,   $x,y \in  \mathcal{X}_{N,1}$ with $x_1=i$ and $y_1=j$,   and assume 
\beq\label{Mijb}
M=M(i,j) : = \lfloor{\tfrac{\abs{i-j}}4}\rfloor= \lfloor{\tfrac{\norm{x-y}}4}\rfloor \ge 8 N+  2 \beta,\mqtx{where}  \beta \;\;\text{is given in  \eq{finiteXXZ}}.
\eeq  
 Fix $s \in (0,1)$, and set $\hat m = \frac m {90}$,  with $m$ as in Lemma~\ref{lem:k=2}, i.e.,
 \beq\label{defhatm}
 \hat m =  \tfrac s {90}   \log \pa{1 + \tfrac {\delta (\Delta-1)}{8 }}.
 \eeq 
 Then  there exists $\what M <\infty$, independent of  $N$, such that for $M\ge \what M$,  outside an event of probability less than $\e^{- \frac {\hat m} 2 M }$   one can decompose $I_{1,\delta} = I_x \cup I_y$,
so that for every  $E \in I_x$ we have
\beq\label{eq:box1} 
  \norm{G_E^{{S_x}}} \leq\e^{\hat m M }\qtx{and} \max_{\substack{w\in  {\La_x}:\\ \norm{w-x}\ge M/2}}  \norm{Q_{i}G_E^{{S_x}}\phi_w} \leq  \e^{-\hat mM}~,
\eeq
which imply
\begin{align}\label{eq:box6} 
\norm{Q_i G_E^{{S_x}} \Gamma_{S_x} }\le \tfrac {1}{\delta (\Delta-1)} M^3 e^{- \hat m M},\end{align}
with similar estimates for $E \in I_y$.  
\end{lemma}

\begin{proof}
The argument will follow closely   \cite[Proof of Proposition 5.1]{ETV}, using  Lemma~\ref{lemsepboxes}.

Fix  $s \in (0,1)$.  Let $E \in I_{1,\delta}$. We first observe that we have
\beq\label{eq:expQ}
\max_{\substack{w\in  {\La_x}:\\ \norm{w-x}\ge \frac M 2}} \E\norm{Q_{i}G_E^{{S_x}}\phi_w}^s \leq \e^{- \frac m {16} M}= \e^{- 10 \hat  m  M},   \qtx{with} \hat m =  \tfrac m {90},
\eeq
for  $M\ge \what M$, where $\what M$ does not depend on $N$ and $m$ is given in \eq{Cms}. The same bound holds if one replaces $x$ with $y$ and $i$ with $j$ in \eq{eq:expQ}. Indeed, introducing $\hat H^{(S_x)}_{N}:=H^{(S_x)}_{N}+P_{1,x}$, and denoting the corresponding resolvent by $\hat G_E^{{S}}$, we obtain, via the resolvent identity, that  for $w\in \La_x$ we have
\beq
\norm{Q_{i}G_E^{{S_x}}\phi_w}^s\le\norm{Q_{i}\hat G_E^{{S_x}}\phi_w}^s+\norm{Q_{i}\hat G_E^{{S_x}}P_{1,x}G_E^{{S_x}}\phi_w}^s.
\eeq
For any $u \in \Lambda_x$ and $v\in S_{\{i\}}$ it is easy to see that $|u-v| \ge\|u-v\| \ge \|x-u\|-(N-1)$, i.e.,  $\dist(u,S_{\{i\}}) \ge \|x-u\|-(N-1)$. Thus, taking expectations on both sides, and using the Combes-Thomas estimate in the form \eqref{eq:fullCT} and Theorem~\ref{cor:dec}, we obtain
\begin{align}\notag
&\E\norm{Q_{i}G_E^{{S_x}}\phi_w}^s\\  \notag & \hskip30pt \le (C^\pr)^s \e^{-m(\norm{x-w} -(N-1))}+(C^\pr)^s \tfrac { C_\tau} {\lambda^s}\sum_{u\in \La_x}\e^{-m \pa{\norm{x-u}-(N-1)}}\e^{-\frac m 2 \norm{u-w}} \\   &\notag  \hskip30pt
 \le  C^s \pa{1 +\tfrac { C_\tau} {\lambda^s}} \e^{mN} \pa{\e^{-m\norm{x-w}} + \sum_{u\in \La_x}\e^{-m \norm{x-u}}\e^{-\frac m 2 \norm{u-w}}}\\   &\notag  \hskip30pt
 \le  C^s \pa{1 +\tfrac { C_\tau} {\lambda^s}} \e^{mN} \pa{\e^{-m\norm{x-w}} + \e^{-\frac m 2 \norm{x-w}}\sum_{u\in \La_x}\e^{-\frac m 2 \norm{x-u}}}
 \\   &\notag  \hskip30pt
 \le  C^s \pa{1 +\tfrac { C_\tau} {\lambda^s}} \e^{mN} \e^{-\frac m 2 \norm{x-w}}\pa{\e^{-\frac m 2 \norm{x-w}} + \sum_{u\in \La_x}\e^{-\frac m 2 \norm{x-u}}}
 \\   &\notag   \hskip30pt\le C^s \pa{1 +\tfrac { C_\tau} {\lambda^s}} \e^{mN}  \e^{-\frac m 2 \norm{x-w}}\pa{1+\tfrac {1 +\e^{-\frac m 2 } }{1- \e^{-\frac m 2 }} }=  \tfrac {2CC^s} {1 -\e^{-\frac m 2 }}\e^{mN}  \e^{-\frac m 2 \norm{x-w}}\\ &  \hskip30pt \le
 C^s \pa{1 +\tfrac { C_\tau} {\lambda^s}} \tfrac {2} {1 -\e^{-\frac m 2 }}\e^{-\frac m8 M} \le \e^{- \frac m {9} M},
\end{align}
where $C^\pr$ is given in \eq{eq:fullCT2},  
$C_\tau$ and $m$ are given in \eq{Cms},   and we used \eq{Mijb} and $\|w-x\|\ge \frac M 2$. The last inequality holds for $M\ge \what M$, where $\what M$ does not depend on $N$ (but depends on $\lambda, \Delta, \delta,s$).

 For, say $x$, define the sets 
\begin{align}
 \Delta^x &:= \set{E \in I_{1,\delta} : \max_{\substack{w\in  {\La_x}:\\ \norm{w-x}\ge M/2}}  \norm{Q_{i}G_E^{{S_x}}\phi_w} > \e^{-\hat m M}}, \nonumber \\
\tilde \Delta^x &:= \set{E \in I_{1,\delta} : \max_{\substack{w\in  {\La_x}:\\ \norm{w-x}\ge M/2}}  \norm{Q_{i}G_E^{{S_x}}\phi_w} >  \e^{- 2 \hat m M}},
\label{eq:deltatilde}
\end{align}
 and  consider the event 
\beq
\tilde B_x= \set{\abs{\tilde \Delta^x} > \e^{-5  \hat mM}},
\eeq
where $\abs{J}$ denotes the Lebesgue measure of the set $J\subset \R$.

If $\tilde{B}_x$ holds, then
\begin{align}
 \int_{I_{1,\delta}}\! \sum_{\substack{w\in  {\La_x}:\\ \norm{w-x}\ge M/2}}\norm{Q_{i}G_E^{{S_x}}\phi_w}^s  \d E
 \geq
\int_{I_{1,\delta}} \max_{\substack{w\in  {\La_x}:\\ \norm{w-x}\ge M/2}}\norm{Q_{i}G_E^{{S_x}}\phi_w}^s > \e^{-5  \hat mM} \e^{-2  \hat mM}>  \e^{-7  \hat m M}.
\end{align}
Using \eqref{eq:expQ} and Chebyshev's inequality, we obtain
\begin{equation}
 \P \{\tilde B_x\} <  M \e^{7  \hat m M}\e^{-10  \hat m M}=M \e^{-3  \hat m M}\le \e^{-2  \hat m M},
\end{equation}
for $M$ large.

We now consider the set
\begin{align}
W_x & = \set{E \in I_{1,\delta} : \dist \pa{E,\sigma({H}_{S_x}^{(L)}) } \le \e^{- \hat m M}}  \notag \\ &   =\set{E \in I_{1,\delta} : \dist \pa{E,\sigma_ {I_{1,\frac \delta 2}}({H}_{S_x}^{(L)}) }  \le  \e^{- \hat m M}},
\end{align}
for sufficiently large $M$. We claim that in the complementary event $\tilde{B}_x^c$ we
have  $ \Delta^x \subset W_x$.  Indeed, suppose $E\in \Delta^x\setminus W_x $. Then there exists $w\in \La_x$ with $ \norm{w-x}\ge M/2$ such
that 
\beq\label{QGS}
\norm{Q_{i}G_E^{{S_x}}\phi_w}> \e^{-\hat m M}.
\eeq
If  $\abs{E-E^\pr} \le 2 \e^{-5  \hat m M}$, we have  
\beq \label{QGS2}
\dist \pa{E^\pr,\sigma({H}_{S_x}^{(L)}) } \ge 
\e^{- \hat m M} -2 \e^{-5  \hat m M} \ge \tfrac 1 2\e^{- \hat m M}.
\eeq
Combining \eq{QGS}  and \eq{QGS2} , we get
\begin{align}
\norm{Q_{i}G_{E^\pr}^{{S_x}}\phi_w}& \ge  \norm{Q_{i}G_E^{{S_x}}\phi_w} - \abs{E-E^\pr} \norm{G_{E}^{{S_x}}}\norm{G_{E'}^{{S_x}}} \notag \\ &  > \e^{-\hat m M} - 4\e^{-5  \hat m M}\e^{2 \hat m M}> \e^{-2\hat m M}.
\end{align}
We infer that $[E-2\e^{-5\hat mM},E+2\e^{-5\hat mM}]\cap I_{1,\delta} \subset  \tilde \Delta^x$, and hence conclude that  $\abs{\tilde \Delta^x} \ge \e^{-5\hat mM}$, and hence the event $\tilde B_x $ holds.   

Thus on $\tilde{B}_x^c$  we have $I_x= I_{1,\delta}\setminus W_x \subset  
 I_{1,\delta}\setminus  \Delta^x$, so \eq{eq:box1} holds for all $E \in I_x$.
Moreover, the same results hold with $y$ substituted for $x$.  To finish the proof we need to estimate the probability of the event
\beq
B_{\rm res}=\set{W_x \cap W_y \ne \emptyset} \subset  \set{\dist \pa{\sigma_ {I_{1,\frac \delta 2}}({H}_{S_x}^{(L)}),\sigma_ {I_{1,\frac \delta 2}}({H}_{S_y}^{(L)}) }  \le 2 \e^{- \hat m M}}.
\eeq
 Since  \eq{Mijb} implies \eq{Mijb9},  it follows from Lemma~\ref{lemsepboxes}(iv) that, for large $M$, we have
\begin{align}
\P\set{B_{\rm res}}\le  2C C_W \lambda^{-1} \|\rho\|_{\infty} M^3 \pa{  \e^{- \hat m M}+   2C_1  e^{-\frac 3 2\eta_1 M}} \le    \e^{- \frac 2 3\hat m M}.
\end{align}

Consider now the event $\cE= \tilde B_x^c \cap \tilde B_y^c \cap B_{\rm res}^c$.  We have
\beq
\P\set{\cE}\ge 1- \pa{ 2 \e^{- 2 \hat m M} + \e^{- \frac 2 3\hat m M}}\ge 1- \e^{-\frac  {\hat m} 2 M},
\eeq
and $  I_{1,\delta}=I_x \cup I_y$ on $\cE$.  

 Now, suppose   \eq{eq:box1}  holds for $E\in I_x$.  We have
\begin{align}
Q_i G_E^{{S_x}} \Gamma_{S_x}  =   Q_i \hat G_E^{{S_x}} \Gamma_{S_x}   +  Q_i  G_E^{{S_x}} P_{1,x}\hat G_E^{{S_x}} \Gamma_{S_x}.
\end{align}
 Since 
 \beq\label{dSS}
 \dist_\infty\pa{S_{\set{i}},\partial_-S_x}\ge \dist_\infty\pa{S_{\set{i}},\partial_+S_x}-1 \ge M,
 \eeq
 and  $\norm{\Gamma_{S_x} }\le \frac N \Delta \le \frac M {8\Delta}$ (recall \eq{GammaS}), it follows from \eq{eq:fullCT} in Lemma~\ref{lem1} that
 \begin{align}
 \norm {Q_i \hat G_E^{{S_x}} \chi_{\partial_- S_x}}\le C^\pr \e^{-\eta^\pr M}, \mqtx{where} C^\pr \sqtx{and} \eta^\pr \mqtx{are given in \eq{eq:fullCT2}}.
 \end{align}
 In addition,
 \begin{align}
 \norm{ Q_i  G_E^{{S_x}} P_{1,x}\hat G_E^{{S_x}} \chi_{\partial_- S_x}}\le \sum_{w\in \Lambda_x} \norm{ Q_i  G_E^{{S_x}}\phi_w} \norm{  \chi_{\partial_- S_x}\hat G_E^{{S_x}}\phi_w}.
 \end{align}
If $w\in \Lambda_x$,  it follows from \eq{eq:box1} and  \eq{eq:fullCT} in Lemma~\ref{lem1}, plus \eq{barHE1} and \eq{defhatm},  that
 \begin{align}
\norm{ Q_i  G_E^{{S_x}}\phi_w} \norm{  \chi_{\partial_- S_x}\hat G_E^{{S_x}}\phi_w} &\le \begin{cases}
\e^{\hat m M}  C^\pr \e^{-\eta^\pr \frac M 2}\le C^\pr \e^{- 44 \hat m M}  & \qtx{if} \norm{x-w}<\frac M2\\
   \e^{- \hat m M}  {2} \pa{\delta \pa{1-\tfrac{1}{\Delta}}}^{-1}        & \qtx{if} \norm{x-w}\ge \frac M2
\end{cases} \notag \\ &   \le  \tfrac {8\Delta}{\delta (\Delta-1)}\e^{- \hat m M} .
 \end{align}
We conclude that 
 \begin{align}
  \norm{ Q_i  G_E^{{S_x}} P_{1,x}\hat G_E^{{S_x}} \chi_{\partial_- S_x}}&\le \tfrac M {8\Delta} (2M+1)  \tfrac {8\Delta}{\delta (\Delta-1)}\e^{- \hat m M}  \le \tfrac {3M^2} {\delta (\Delta-1)}\e^{- \hat m M}.
 \end{align}

Combining the estimates, and taking $\hat M$ sufficiently large, we get
\begin{align}
\norm{Q_i G_E^{{S_x}} \Gamma_{S_x} }&\le  \pa{\tfrac {8\Delta}{\delta (\Delta-1)} \e^{- \eta^\pr M} + \tfrac {3M^2} {\delta (\Delta-1)}\e^{- \hat m M}}\norm{\Gamma_{S_x}}\le \tfrac {1}{\delta (\Delta-1)} M^3 e^{- \hat m M} .
\end{align}

The lemma is proven.
\end{proof}

\section{Eigencorrelators in the droplet spectrum}  \label{sec:ecorloc} 

In view of the equivalence of \eq{eq:efcor5} and \eq{eq:efcor}, Theorem~\ref{thm:efcor} follows from the following theorem.

\begin{theorem}\label{thm:efcor2} Fix $s\in (0,1)$, and assume  $\Delta >1$ and $\lambda >0$ satisfy the condition \eq{lambdaDeltahyp} of Theorem~\ref{cor:dec}. Let $Q_N^{(L)}(i,j;I_{1,\delta})$ be the eigenfunction correlators defined in \eq{eq:Q_I}. 
  There exist constants 
$C<\infty$ and $m>0$ such that
  \beq \label{eq:efcor2}
\sum_{N=1}^{\infty} \E\pa{Q_N^{(L)}(i,j;I_{1,\delta})}\le C e^{-m|i-j|} \qtx{for all} -L \le i, j \le L,
\eeq 
uniformly in $L$. 
\end{theorem}

In the context of many-body localization,  $Q_N^{(L)}(i,j;I)$ is the relevant eigenfunction correlator for  $H^{\pa{L}}_N$  corresponding to a finite interval $I\subset \mathbb{R}$ and a pair of indices $i,j\in\Z$.   Our proof   follows the outlines of  the proofs of dynamical localization for random Schr\"odinger operators given in \cite{GKboot,GKjsp,GKber,Kl}.
The key  ingredients are Lemma~\ref{lem:pairb}   and the a priori bound on eigenfunction correlators derived in Lemma~\ref{lem:PQ} below.

 For a given $N\in \N$, let $P_1$ be the indicator function of $\mathcal X_{N,1}^{(L)}:= \mathcal X_{N,1}\cap \mathcal{X}_N^{(L)}$, and set  $\hat{H}_N^{(L)}=H_N^{(L)}+P_1$. Note that, due to working in finite volume, all spectra are finite, and that $\hat{H}_N^{(L)}$ has no spectrum below $2(1-\frac 1\Delta)$. Let $G_z$ and $\hat G_z$ be the corresponding resolvents.

Let  $i, j \in \Z$ and  $x,y\in \mathcal{X}_{N,1}$ with  $x_1=i$ and $y_1=j$.  Let $Q_{i}$ and $Q_{j}$  be the indicator functions of the subsets $S_{\set{i}}\cap \mathcal{X}_N^{(L)}$ and $S_{\set{j}}\cap \mathcal{X}_N^{(L)}$, respectively, defined in \eqref{eq:setLa'}. 
Given an energy interval $I$, set  $\sigma_{I}=\sigma \pa{H^{\pa{L}}_N} \cap I$. 
We have
\beq
Q_N^{(L)}(k,k;I) = \sum_{E \in \sigma_I}\norm{Q_kP_{E}Q_k}_1 =\norm{Q_k\,P_{I}Q_k}_1 \qtx{for all} k \in \Z.
\eeq
Moreover
\begin{align}\notag
&\sum_{E \in \sigma_{I}} \norm{Q_iP_{E}Q_j}_1  \le  \sum_{E \in \sigma_{I}} \norm{Q_iP_{E}}_2 \norm{Q_jP_{E}}_2= \sum_{E \in \sigma_I}\norm{Q_iP_{E}Q_i}_1^{\frac 12} \norm{Q_jP_{E}Q_j}_1^{\frac 12}\\   &\quad  \le \pa{\sum_{E \in \sigma_I}\norm{Q_iP_{E}Q_i}_1}^{1/2}\pa{\sum_{E \in \sigma_I}\norm{Q_jP_{E}Q_j}_1}^{1/2}=\norm{ Q_i P_{{I}}Q_i}_1^{1/2}\norm{ Q_j P_{{I}}Q_j}_1^{1/2}.
\label{QQ12}
\end{align} 
 Thus  it follows from \eq{eq:Q_I} and \eq{QQ12} that
 that 
\begin{align}
Q_N^{(L)}(i,j;I)  \le \sqrt{Q_N^{(L)}(i,i;I) Q_N^{(L)}(j,j;I)}.
\end{align}

The following lemma provides an a priori bound  for the interval $I_{1,\delta}$.
\begin{lemma}\label{lem:PQ} 
For all  $i\in \Z$ and $N\in\N$ we have  
\beq \label{QPQ4}
Q_N^{(L)}(i,i;I_{1,\delta})\le \norm{ P_{I_{1,\delta}}Q_i}_1\le \check C N, \qtx{where}  \check C= \tfrac {16\Delta} {\delta(\Delta -1)}\pa{\tfrac{8 } {\delta (\Delta-1)}+2}.
\eeq 
In particular, 
\beq\label{apQ}
 Q_N^{(L)}(i,j;I_{1,\delta}) \le \check C N.
\eeq  
\end{lemma}

\begin{proof}
Let $\Gamma$ be the circle of radius $\frac 1 2(1-\frac 1\Delta)$ in $\C$, centered at the midpoint of $I_{1,\delta}$, so that $\Gamma$ encloses $I_{1,\delta}$ and $\dist\pa{z, I_{1,\delta}}$ as well as $\dist(z,\sigma(\hat{H}_N^{(L)})$ are both at least $\frac \delta  2(1-\frac 1\Delta)$ for any $z\in\Gamma$.  We obtain, using $\dist\pa{I_{1,\delta},\sigma\pa{\hat H^{\pa{L}}_N}}\ge \delta(1-\tfrac 1\Delta)$,
\beq
P_{I_{1,\delta}}=\frac{1}{2\pi i}\,P_{I_{1,\delta}}\,\oint_\Gamma G_z \d z=\frac{1}{2\pi i}\,P_{I_{1,\delta}}\,\oint_\Gamma \pa{G_z-\hat G_z} \d z=\frac{1}{2\pi i}\,P_{I_{1,\delta}}\,\oint_\Gamma G_z \,P_{1}\, \hat G_z\,\d z.
\eeq  
Therefore,
 \beq
\norm{ P_{I_{1,\delta}}Q_i}_1\le \tfrac{1}{2}\left(1-\tfrac{1}{\Delta}\right) \max_{z\in\Gamma}\norm{ P_{I_{1,\delta}}G_z}\norm{P_1\hat G_zQ_i}_1\le \tfrac{1}{\delta} \max_{z\in\Gamma} \norm{P_1\hat G_zQ_i}_1.
\eeq 
For $z\in \Gamma$ we have the norm bound $\norm{\hat G_{z}}\le\frac  2{\delta(1-\frac 1\Delta)}$.  Also, for $u \in \mathcal X_{N,1}^{(L)}$ we have $\dist(u,S_{\{i\}}) \ge \|x-u\|-(N-1)$ and  thus it follows from the Combes-Thomas bound \eq{eq:fullCT} that
\beq
\norm{Q_i \hat G_{\bar{z}} \phi_u}\le C^\pr\min\pa{1,\e^{-\eta^\pr \pa{\norm{x-u}-(N-1)}}},
\eeq
where $C^\pr$ and $\eta^\pr$ are given in \eq{eq:fullCT2}.  Here $x=(i,\ldots,i+N-1)$, as in Section~\ref{sec:pairofboxes}. We conclude that
\begin{align}\notag
\norm{P_1\hat G_zQ_i}_1 & \le \sum_{u\in \mathcal X_{N,1}^{(L)}}\norm{Q_i \hat{G}_{\bar{z}} \phi_u}\le C^\pr\pa{ \sum_{\norm{x-u}>(N-1)}\e^{-\eta^\pr\pa{\norm{x-u}-(N-1)}}+2N}
\\ & \label{eq:sumb} \le
 2 C^\pr\pa{\pa{\e^{\eta^\pr}-1}^{-1} +N}= 2 \tfrac {8\Delta} {\delta(\Delta -1)}\pa{\tfrac{8 } {\delta (\Delta-1)}+N} ,  
\end{align} 
which yields \eq{QPQ4}.
\end{proof}

\begin{proof}[Proof of Theorem~\ref{thm:efcor2}]

Let $E\in I_{1,\delta}$. We have  $\pa{H^{\pa{L}}_N -E}P_E=0$.  Note that if  $i\in \Z$, $x\in \mathcal{X}_{N,1}$ with  $x_1=i$, $M \ge 2N$, $S_x=S_M(x)$, and $E \notin \sigma (H^{(L)}_{S_x})$, we have 
\begin{align}\label{EIN}
 Q_i P_E&= -  Q_i G_E^{{S_x}} \Gamma_{S_x}  P_E.
\end{align}

We now suppose  \eq{Mijb} is satisfied,  so we can use Lemma~\ref{lem:pairb}. Thus, for $M\ge \what M$, outside an event of probability $\le \e^{- \frac {\hat m} 2 M }$, we have $I_{1,\delta} = I_x \cup I_y$, which can be chosen so $ I_x \cap I_y=\emptyset$, and 
          \eq{eq:box1}  and \eq{eq:box6} hold for $E\in I_x$, with similar inequalities for $E\in I_y$. Take $E \in I_x$. Then $E \notin \sigma (H^{(L)}_{S_x})$, and \eq{EIN} holds.  

Let  $k\in \Z$ and  $u\in \mathcal{X}_{N,1}$ with  $u_1=k$, where we allow for  $k=i$, and let   $Q_{k}$   be the indicator function of the set $S_{\set{k}}\cap \mathcal{X}_N^{(L)}$. 
Then, using \eq{EIN} and \eq{eq:box6}, we get
\begin{align}\label{QPEQj}
\norm{Q_i\,P_{E}\,Q_k}_1\le 
\norm{Q_i G_E^{{S_x}} \Gamma_{S_x} }\norm{\chi_{\partial_+S_x}P_{E}Q_k}_1\le
 \tfrac {1}{\delta (\Delta-1)} M^3 e^{- \hat m M}\norm{\chi_{\partial_+S_x}P_{E}Q_k}_1.
\end{align}
Since \  $\partial_+S_x\subset   S_{\set{i-M -1}}\cup S_{\set{i+M +1}}   $, we have 
 $\chi_{\partial_+S_x} \le Q_{{i-M -1}} +  Q_{{i+M +1}}$, and hence
  \begin{align}\label{partialPQ}
 \norm{\chi_{\partial_+S_x}P_{E}Q_k}_1 &\le \norm{\pa{Q_{{i-M -1}} +  Q_{{i+M +1}}}P_{E}Q_k}_1  \le
 \norm{Q_{{i-M -1}} P_{E}Q_k}_1 +\norm{Q_{{i+M +1}} P_{E}Q_k}_1 \notag \\ &  \le
\norm{ P_{E}Q_k}_2\pa{\norm{Q_{{i-M -1}} P_{E}}_2 +\norm{Q_{{i+M +1}} P_{E}}_2} \notag \\ &   \le
 {\norm{ P_{E}Q_k}_2^2 + \tfrac 12 \pa{\norm{ P_{E}Q_{{i-M -1}}}_2^2 +\norm{ P_{E}Q_{{i+M +1}}}_2^2}} .
 \end{align}
 We conclude, using Lemma~\ref{lem:PQ}, that
 \begin{align}
  \sum_{E \in \sigma_{I_x}}  \norm{\chi_{\partial_+S_x}P_{E}Q_k}_1\le 2 \check C N\le \tfrac {\check C} 4 M.
 \end{align}
 Combining \eq{QPEQj} and \eq{partialPQ} we get
 \begin{align}\label{sumIx}
 \sum_{E \in \sigma_{I_x}} \norm{Q_i\,P_{E}\,Q_k}_1\le\tfrac {\check  C}{4\delta (\Delta-1)} M^4 e^{- \hat m M}.
  \end{align}

 Since the same estimates hold for $E\in I_y$, with $Q_j$ substituted for $Q_i$, we conclude that outside an event of probability $\le \e^{- \frac {\hat m} 2 M }$, we have 
 \begin{align}\label{Qij3}
 Q_N^{(L)}(i,j;I_{1,\delta})=\sum_{E \in \sigma_{I_x}}\norm{Q_i\,P_{E}\,Q_j}_1 +\sum_{E \in \sigma_{I_y}}\norm{Q_j\,P_{E}\,Q_i}_1
\le
 \tfrac {\check  C}{2\delta (\Delta-1)} M^4 e^{- \hat m M}\le C_1 e^{- \frac {\hat m} 2 M},
 \end{align} 
where the last inequality holds for $M$ large. 
 
 Using  \eq{apQ} in the event of probability $\le \e^{- \frac {\hat m} 2 M }$,  we conclude that for   $8N+ 2\beta\le M$, where $M $ is sufficiently large, 
 \begin{align}\label{EQij1}
 \E\pa{Q_N^{(L)}(i,j;I_{1,\delta})}\le  C_1  e^{- \frac {\hat m} 2 M} + \tfrac {\check C} 8  M \e^{- \frac {\hat m} 2 M }\le  C_2e^{- \frac {\hat m} 9 \abs{i-j}}.
 \end{align} 
In particular, for sufficiently large $M$, say   $M\ge \what {\what M} $,  we have
\begin{align}\label{EQij2}
\sum_{N=1}^{\fl{\frac {\abs{i-j}}{32}- \frac { 2\beta +1}8} }\E\pa{Q_N^{(L)}(i,j;I_{1,\delta})}\le C_2 \tfrac {\abs{i-j}}{32} e^{- \frac {\hat m} 9 \abs{i-j}}\le C_3 e^{- \frac {\hat m} {10} \abs{i-j}}.
\end{align}

Now let   $N >    N_1$ for some $N_1\in \N$.  Let $\bar{\mu}= \E \set{\omega_0}$, and assume $N \bar{\mu} > 2$.  Then the standard large deviations estimate
gives
\begin{align}
\P\set{V_\omega (x)< 1}\le \P\set{V_\omega (x)< N \tfrac {\bar{\mu}} 2}\le \e^{- c_\mu N},
\end{align}
where $c_\mu$ is a constant depending only on the probability distribution $\mu$. It implies that, setting $C_\mu= \e^{ 2c_\mu /\bar \mu}$, we have $\P\set{V_\omega (x)< 1}\le  C_\mu\e^{- c_\mu N}$ for all $N\in \N$.
  If
  $i\in \Z$, $x\in \mathcal{X}_{N,1}$ with  $x_1=i$, taking $M =2 N$, $S_x=S_M(x)$,  then
  \begin{align}\label{PV<}
& \P_{\set{i-M-N+1,i-M-N+2,\ldots,i+M+N-1 } }\set{V_\omega (x)< 1 \sqtx{for some} x \in \Lambda_x} \notag \\ & \hskip100pt  \le \pa{2M+1} C_\mu \e^{- c_\mu N}= C_\mu\pa{4N+1}  \e^{- c_\mu N}.
  \end{align}
  Thus, outside an event of probability less than $C_\mu\pa{4N+1} \e^{- c_\mu N}$ we have
  \beq\label{PVP}
  P_{1,x} V_\omega  P_{1,x}  \ge  P_{1,x} ,
  \eeq
 and  hence, in  particular, 
  $E \notin \sigma (H^{(L)}_{S_x})$.  Moreover, in view of \eq{PVP}, the Combes-Thomas estimate \eq{eq:fullCT} in Lemma~\ref{lem1}  holds for $H^{(L)}_{S_x}$.  Thus 
  \begin{align}
  \norm{Q_i G_E^{{S_x}} \Gamma_{S_x} }\le  \norm {Q_i \ G_E^{{S_x}} \chi_{\partial_- S_x}}\norm{\Gamma_{S_x}} \le \tfrac N {\Delta} C^\pr \e^{-\eta^\pr M}\le  \tfrac {8N} {\delta (\Delta-1)}  \e^{-180 \hat m  N},
  \end{align}
 where $C^\pr$ and $\eta^\pr$ are given in \eq{eq:fullCT2}, and we used  \eq{dSS} and \eq{defhatm}. Thus, using \eq{EIN}, and proceeding as in \eq{Qij3}, we get
  \begin{align}\label{Qij99}
  Q_N^{(L)}(i,j;I_{1,\delta})\le \tfrac {16 \check C} {\delta (\Delta-1)} N^2 \e^{-180 \hat m  N}.
 \end{align}
  It follows, once again using Lemma~\ref{lem:PQ}, and \eq{PV<},  that
  \begin{align}
 \E\pa{Q_N^{(L)}(i,j;I_{1,\delta})}&\le \tfrac {16 \check C} {\delta (\Delta-1)} N^2 \e^{-180 \hat m  N} + \check CC_\mu N  \pa{4N+1}  \e^{- c_\mu N} \le \tilde C e^{-  {\tilde m}  N},
 \end{align}
where   $\tilde m= \frac 1 2\min \set{c_\mu, 180 \hat m }$, and $\tilde C$ is a constant independent of $N$.  We conclude that 
\begin{align}\label{EQij3}
\sum_{N=N_1+1}^\infty  \E\pa{Q_N^{(L)}(i,j;I_{1,\delta})}\le\sum_{N=N_1+1}^\infty  \tilde C e^{-  {\tilde m}  N}\le C_{4}\, e^{- \frac {\tilde m} {2} N_1}.
 \end{align}

If $\abs{i-j}>4 \pa{\what{\what M} +1}$, we take $N_1=\fl{\frac {\abs{i-j}}{32}- \frac { 2\beta +1}8} $, obtaining
\begin{align}\label{EQij399}
\sum_{N={\fl{\frac {\abs{i-j}}{32}- \frac { 2\beta +1}8} }+1}^\infty  \E\pa{Q_N^{(L)}(i,j;I_{1,\delta})}\le\ C_{4}\, e^{- \frac {\tilde m} {70} \abs{i-j}}.
 \end{align} 
 Combining \eq{EQij2} and \eq{EQij399}, we get  that for $\abs{i-j}>4 \pa{\what{\what M} +1}$ we have
  \begin{align}\label{EQij4}
\sum_{N=1}^\infty  \E\pa{Q_N^{(L)}(i,j;I_{1,\delta})}\le C_{5}e^{-  \check{m}\abs{i-j}},
 \end{align}
  for all $-L\le i,j \le L\in \Z$, with   $\check{m}= \min\set{\frac {\hat m}{10},\frac {\tilde m} {70}} $,
with $C_5$ a constant . 

Finally, we consider the case when $\abs{i-j}\le4 \pa{\what{\what M}  +1}$.  In this case we take $N_1= \abs{i-j}$, and  using  Lemma~\ref{lem:PQ} and \eq{EQij3}, we conclude that 
\begin{align}\label{EQij386}
&\sum_{N=1}^\infty  \E\pa{Q_N^{(L)}(i,j;I_{1,\delta})}\le \check C  \abs{i-j} ^2 + C_{4}\, e^{- \frac {\tilde m} {2}  \abs{i-j} } \notag \\   &  \hskip70pt \le 16 \check C \pa{\what{\what M}  +1}^2 + C_{4}\, e^{- \frac {\tilde m} {2}  \abs{i-j} }\le C_6 \,  e^{- \frac {\tilde m} {2}  \abs{i-j} } ,
 \end{align}
for some constant $C_6$.

The proof is complete.
\end{proof}

\end{document}